\documentclass[twocolumn]{IEEEtran}

\usepackage{amsmath,epsfig,amssymb,verbatim,amsopn,cite,multirow}
\usepackage{amsthm}
\usepackage{balance}
\usepackage{multirow}
\usepackage[usenames,dvipsnames]{color}
\usepackage[all]{xy}  
\usepackage{url}
\usepackage{amsfonts}
\usepackage{amssymb}
\usepackage{epsfig}
\usepackage{epstopdf}
\usepackage{bm}
\usepackage{balance}
\usepackage{graphicx}
\usepackage{subcaption}
\usepackage{footnote}
\usepackage{cancel}
\usepackage{algorithm,algorithmic}

\usepackage{multicol} 
\usepackage{makecell} 
\usepackage[short]{optidef} 

\usepackage{tikz}
\newcommand{\tikzxmark}{%
\tikz[scale=0.15] {
    \draw[line width=0.6,line cap=round] (0,0) to [bend left=5] (1,1);
    \draw[line width=0.6,line cap=round] (0.1,0.85) to [bend right=3] (0.8,0.05);
}}

\newtheorem{Lemma}{Lemma}

\newtheorem{Remark}{Remark}

  {\proof}{\proofend}
\newtheorem{proposition}{Proposition}

\DeclareMathOperator{\OO}{\mathcal{O}}
\newcommand{\qa}{{\bf a}}

\newcommand{\qe}{{\bf e}}

\newcommand{\qg}{{ \textbf{g} }}
\newcommand{\qh}{{ \textbf{h} }}

\newcommand{\qn}{{\bf n}}

\newcommand{\qs}{{\bf s}}

\newcommand{\qw}{{\bf w}}
\newcommand{\qx}{{\bf x}}
\newcommand{\qy}{{ \textbf{y} }}
\newcommand{\qz}{{ \textbf{z} }}

\newcommand{\qA}{{\bf A}}
\newcommand{\qB}{{\bf B}}

\newcommand{\qF}{{\bf F}}
\newcommand{\qG}{{ \textbf{G} }}
\newcommand{\qH}{{ \textbf{H} }}
\newcommand{\qI}{{ \textbf{I} }}

\newcommand{\qN}{{\bf N}}

\newcommand{\qR}{{\bf R}}

\newcommand{\qU}{{\bf U}}
\newcommand{\qV}{{\bf V}}

\newcommand{\qY}{{\bf Y}}
\newcommand{\qZ}{{ \textbf{Z} }}

\DeclareMathOperator*{\argmax}{arg\,max}
\newcommand{\ettall}{\emph{et al.}}

\newcommand{\UE}{\mathtt{I}}
\newcommand{\sn}{\mathtt{E}}

\DeclareMathOperator{\ETAI}{\boldsymbol{\eta}^{\mathtt{I}}}
\DeclareMathOperator{\ETAE}{\boldsymbol{\eta}^{\mathtt{E}}}
\DeclareMathOperator{\aaa}{\mathbf{a}}

\DeclareMathOperator{\MM}{\mathcal{M}}

\DeclareMathOperator{\K}{\mathcal{K}}

\DeclareMathOperator{\J}{\mathcal{J}}
\DeclareMathOperator{\M}{\mathcal{M}}

\DeclareMathOperator{\C}{\mathbb{C}}

\DeclareMathOperator{\CN}{\mathcal{CN}}

\newcommand{\PZF}{\mathsf{PZF}}
\newcommand{\PMRT}{\mathsf{PMRT}}

\newcommand{\wimk}{\qw_{\mathrm{I},mk}}

\newcommand{\wemj}{\qw_{\mathrm{E},mj}}

\newcommand{\SEth}{\mathcal{S}_{k}}

\newcommand{\Pir}{\mathcal{P}^{\mathtt{I}}_{k}}
\newcommand{\Per}{\mathcal{P}^{\mathtt{E}}_{j}}

\newcommand{\wimkp}{\qw_{\mathrm{I},mk'}}

\newcommand{\wemjp}{\qw_{\mathrm{E},mj'}}
\newcommand{\Ghms}{\hat{\qG}_m^{\sn}}

\newcommand{\Ghmu}{\hat{\qG}_m^{\UE}}

\newcommand{\Snn}{\sigma_n^2}
\newcommand{\Ex}{\mathbb{E}}

\newcommand{\yej}{y_{\mathtt{E},j}}

\newcommand{\yik}{y_{\mathtt{I},k}}

\newcommand{\gmkiu}{\qg_{mk}^{\UE}}

\newcommand{\gmjue}{\qg_{mj}^{\sn}}

\newcommand{\gmjpue}{\qg_{mj'}^{\sn}}

\newcommand{\hgmjue}{\hat{\qg}_{mj}^{\sn}}
\newcommand{\tildehgmjue}{\tilde{\hat{\qg}}_{mj}^{\sn}}

\newcommand{\hmlue}{\qh_{mj}^{\sn}}
\newcommand{\Hmlue}{\qH_{m}}

\newcommand{\gmjeulos}{\bar{\qg}^{\mathrm{E}}_{mj}}
\newcommand{\gmjpeulos}{\bar{\qg}^{\mathrm{E}}_{mj'}}

\newcommand{\FmRIS}{\qF_{m}}

\newcommand{\FmRISlos}{\bar{\qF}_{m}}
\newcommand{\FmRISnlos}{\tilde{\qF}_{m}}
\newcommand{\zmjlos}{\bar{\qz}_{mj}}

\newcommand{\trace}{\mathrm{tr}}

\newcommand{\gmkpue}{\qg_{mk'}^{\UE}}

\newcommand{\hgmkue}{\hat{\qg}_{mk}^{\UE}}

\newcommand{\ghmkue}{\hat{\qg}_{mk}^{\UE}}
\newcommand{\gtilmkue}{\tilde{\qg}_{mk}^{\UE}}

\newcommand{\gtilmjeu}{\tilde{\qg}_{mj}^{\sn}}
\newcommand{\gtilhatmjeu}{\tilde{\hat{\qg}}_{mj}^{\sn}}
\newcommand{\gtilhatmjpeu}{\tilde{\hat{\qg}}_{mj'}^{\sn}}
\newcommand{\gamuemk}{\gamma_{mk}^{\UE}}

\newcommand{\gameumj}{\gamma_{mj}^{\sn}}
\newcommand{\gameumjp}{\gamma_{mj'}^{\sn}}
\newcommand{\betamkue}{\beta_{mk}^{\UE}}
\newcommand{\betamjeu}{\beta_{mj}^{\sn}}
\newcommand{\betamjpeu}{\beta_{mj'}^{\sn}}

\newcommand{\betamkpue}{\beta_{mk'}^{\UE}}

\newcommand{\etamkIn}{\eta_{mk}^{\mathrm{I} (n)}}

\newcommand{\etamkI}{\eta_{mk}^{\mathtt{I}}}

\newcommand{\etamkpI}{\eta_{mk'}^{\mathtt{I}}}
\newcommand{\etamjE}{\eta_{mj}^{\mathtt{E}}}

\newcommand{\etamjpE}{\eta_{mj'}^{\mathtt{E}}}

\newcommand{\etamI}{\eta_{m}^{\mathtt{I}}}
\newcommand{\etamE}{\eta_{m}^{\mathtt{E}}}

\newcommand{\SEk}{\mathrm{SE}_k}
\newcommand{\SINRk}{\mathrm{SINR}_k}

\newcommand{\HEU}{\textbf{HEU-BD-RIS}}
\newcommand{\DFT}{\textbf{DFT-BD-RIS}}

\newcommand{\xik}{x_{\mathtt{I},k}}
\newcommand{\xikp}{x_{\mathtt{I},k'}}
\newcommand{\xej}{x_{\mathtt{E},j}}
\newcommand{\xejp}{x_{\mathtt{E},j'}}

\DeclareMathOperator{\PSI}{\boldsymbol{\Psi}}

\DeclareMathOperator{\THETA}{\boldsymbol{\Theta }}

\DeclareMathOperator{\LAMBDA}{\boldsymbol{\Lambda}}

\DeclareMathOperator{\VARPHI}{\boldsymbol{\varphi}}

\begin{document}

\title{Cell-Free Massive MIMO SWIPT with Beyond Diagonal Reconfigurable Intelligent Surfaces}
\author{Thien Duc Hua,~\IEEEmembership{Graduate Student Member,~IEEE,}, Mohammadali Mohammadi,~\IEEEmembership{Senior Member,~IEEE,}
\\
Hien Quoc Ngo,~\IEEEmembership{Fellow,~IEEE,}   and  Michail Matthaiou,~\IEEEmembership{Fellow,~IEEE}
\thanks{
This work was supported by the U.K. Engineering and Physical Sciences Research
Council (EPSRC) (grant No. EP/X04047X/1). The work of  H. Q. Ngo was supported by the U.K. Research and Innovation Future Leaders Fellowships under Grant MR/X010635/1, and a research grant from the Department for the Economy Northern Ireland under the US-Ireland R\&D Partnership Programme. The work of T. D. Hua and M. Matthaiou was supported by the European Research Council (ERC) under the European Union’s Horizon 2020 research
and innovation programme (grant agreement No. 101001331).}
\thanks{The authors are with the Centre for Wireless Innovation (CWI), Queen's University Belfast, BT3 9DT Belfast, U.K.
(email:\{dhua01, m.mohammadi, hien.ngo, m.matthaiou\}@qub.ac.uk).}
\thanks{ Parts of this paper appeared at the 2024 IEEE WCNC~\cite{Hua:WCNC:2024}.
}}

\markboth{.}%
{Shell \MakeLowercase{\textit{et al.}}: A Sample Article Using IEEEtran.cls for IEEE Journals}

\maketitle

\begin{abstract}
We investigate the integration of beyond-diagonal reconfigurable intelligent surfaces (BD-RISs) into cell-free massive multiple-input multiple-output (CF-mMIMO) systems to enhance simultaneous wireless information and power transfer (SWIPT). To simultaneously support two groups of users—energy receivers (ERs) and information receivers (IRs)— without sacrificing time-frequency resources, a subset of access points (APs) is dedicated to serving ERs with the aid of a BD-RIS, while the remaining APs focus on supporting IRs. A protective partial zero-forcing precoding technique is implemented at the APs to manage the non-coherent interference between the ERs and IRs. Subsequently, closed-form expressions for the spectral efficiency of the IRs and the average sum of harvested energy (HE) at the ERs are leveraged to formulate a comprehensive optimization problem. This problem jointly optimizes the AP selection, AP power control, and scattering matrix design at the BD-RIS, all based on long-term statistical channel state information. This challenging problem is then effectively transformed into more tractable forms. To solve these sub-problems, efficient algorithms are proposed, including a heuristic search for the scattering matrix design, as well as successive convex approximation and deep reinforcement learning methods for the joint AP mode selection and power control design. Numerical results show that a BD-RIS with a group- or fully-connected architecture achieves significant EH gains over the conventional diagonal RIS, especially delivering up to a $7$-fold increase in the average sum of HE when a heuristic-based scattering matrix design is employed.
\end{abstract}


\vspace{-1em}
\section{Introduction}
The limited lifespan of Internet of Things (IoT) devices/sensors poses a significant challenge to their large-scale deployment. To overcome these energy limitations, it is crucial to explore alternative energy sources that can extend battery life and support the long-term sustainability of these devices. As a result, there is an urgent need for cutting-edge solutions to efficiently recharge the batteries of countless sensors, ensuring reliable operation across extensive IoT networks. In this regard, cellular BSs can be utilized to simultaneously support two groups of receivers: 1) information receivers (IRs), representing conventional mobile users, and 2) energy receivers (ERs), which are the sensors and IoT devices. This configuration is known as simultaneous wireless information and power transfer (SWIPT)~\cite{Alsaba:TuT:2018}. However, implementing this topology faces a significant bottleneck due to the disparity in power supply requirements between the two groups. Specifically, ERs require much higher power signals to compensate for path loss and meet the requirements of their non-linear energy harvesting (EH) circuits. However, increasing the transmit power towards ERs can amplify interference, potentially disrupting the operation of IRs. One preliminary approach is to reduce the distance between the BS and ERs. However, conventional cellular networks, which are primarily designed for serving IRs, do not facilitate this adjustment.

To address these challenges, a practical approach is to utilize cell-free massive multiple-input multiple-output (CF-mMIMO) technology, which represents an advanced evolution of distributed and network MIMO systems~\cite{cite:HienNgo:cf02:2018, mohammadi2024next}. In this architecture, a large number of access points (APs) are deployed throughout a coverage area to simultaneously serve many users. These APs work collaboratively through a fronthaul network, eliminating cell boundaries and enabling them to be positioned closer to the IRs and ERs, thereby facilitating SWIPT applications. Moreover, the deployment of reconfigurable intelligent surfaces (RISs), which offer notable advantages in EH and SWIPT applications, can significantly enhance the overall EH performance~\cite{cite:reviewer3a, cite:reviewer3b, Kaixi:2024:WCL, Kaixi:2021:China, Mohammadi:TC:2024, cite:R1}.
RISs can improve the channel gain between a power source and the potential ERs in harsh propagation environments, while also addressing coverage gaps and blind spots. A recent advancement in this area is the concept of beyond-diagonal RIS (BD-RIS), which consolidates various RIS modes and architectures. According to~\cite{cite:HongyuLi:BDRISoverview01:2023, Li:JSAC:2023}, a BD-RIS differs from conventional diagonal RIS (D-RIS) by using non-diagonal scattering matrices enabled by inter-port tunable admittances, allowing advanced beamforming via interconnected elements. These matrices provide greater flexibility and performance potential than diagonal designs. However, BD-RIS has not yet been explored in SWIPT scenarios. Motivated by this, we investigate its integration into SWIPT CF-mMIMO networks.
\begin{table*}
	\centering
	\caption{\label{tabel:Survey} Contrasting our contributions to the  CF-mMIMO SWIPT literature}
	\vspace{-0.6em}
	\small
\begin{tabular}{|m{2.4cm}|>{\centering\arraybackslash}m{0.8cm}|
>{\centering\arraybackslash}m{0.40cm}|
>{\centering\arraybackslash}m{0.40cm}|
>{\centering\arraybackslash}m{0.55cm}|
>{\centering\arraybackslash}m{0.55cm}|
>{\centering\arraybackslash}m{0.55cm}|
>{\centering\arraybackslash}m{0.55cm}|
>{\centering\arraybackslash}m{0.55cm}|
>{\centering\arraybackslash}m{0.55cm}|
>{\centering\arraybackslash}m{0.55cm}|
>{\centering\arraybackslash}m{0.55cm}|
>{\centering\arraybackslash}m{0.55cm}|
>{\centering\arraybackslash}m{0.55cm}|
>{\centering\arraybackslash}m{0.55cm}|
>{\centering\arraybackslash}m{0.55cm}|}

	\hline
        \centering\textbf{Contributions} 
        &\centering \textbf{This paper}
        &\centering\cite{Kaixi:2024:WCL}
        &\centering\cite{Kaixi:2021:China}
        &\centering\cite{Wang:JIOT:2020}
        &\centering\cite{Demir:TWC:2021}
        &\centering\cite{Femenias:TCOM:2021}
        &\centering\cite{Xinjiang:TWC:2021}
        &\centering\cite{Zhang:IoT:2022}
        &\centering\cite{Galappaththige:WCL:2024}
        &\centering\cite{Zhang:TWC:2023}
        &\centering\cite{Mohammadi:GC:2023}
        &\cite{cite:Chien:TWC:2022}
        &\cite{Lan:IoT:2024}
        &\cite{Elhoushy:WCL:2022}
        \cr

        \hline

        D-RIS     
        &\centering\checkmark 
        &\centering\checkmark 
        &\centering\checkmark 
        &\centering\tikzxmark 
        &\centering\tikzxmark      
        &\centering\tikzxmark 
        &\centering\tikzxmark 
        &\centering\tikzxmark 
        &\centering\tikzxmark 
        &\centering\tikzxmark 
        &\centering\tikzxmark 
        &\centering\checkmark
        &\centering\checkmark
        &\centering\checkmark
        \cr        
        \hline

        BD-RIS     
        &\centering\checkmark 
        &\centering\tikzxmark 
        &\centering\tikzxmark 
        &\centering\tikzxmark   
        &\centering\tikzxmark      
        &\centering\tikzxmark
        &\centering\tikzxmark
        &\centering\tikzxmark
        &\centering\tikzxmark
        &\centering\tikzxmark
        &\centering\tikzxmark
        &\centering\tikzxmark
        &\centering\tikzxmark
        &\centering\tikzxmark
        \cr
        
       \hline
       
        SWIPT         
        &\centering\checkmark
        &\centering\checkmark 
        &\centering\checkmark 
        &\centering\checkmark   
        &\centering\checkmark 
        &\centering\checkmark 
        & \centering\checkmark
        & \centering\checkmark
        & \centering\checkmark 
        & \centering\checkmark
        & \centering\checkmark
        &\centering\tikzxmark
        &\centering\tikzxmark
        &\centering\tikzxmark
        \cr

        \hline

       NL-EH model         
        &\centering\checkmark 
        &\centering\tikzxmark
        &\centering\tikzxmark
        &\centering\tikzxmark   
        &\centering\checkmark 
        &\centering\tikzxmark 
        &\centering\tikzxmark
        & \centering\checkmark
        & \centering\checkmark 
        &\centering\tikzxmark
        &\centering\tikzxmark
        &\centering\tikzxmark
        &\centering\tikzxmark
        &\centering\tikzxmark
        \cr

        \hline
        Statistical CSI         
        &\centering\checkmark
        &\centering\checkmark 
        &\centering\tikzxmark
        &\centering\checkmark   
        &\centering\checkmark 
        &\centering \checkmark  
        &\centering\checkmark 
        &\centering\checkmark 
        &\centering\tikzxmark
        &\centering\tikzxmark
        &\centering\tikzxmark
        &\centering\checkmark
        &\centering\tikzxmark
        &\centering\checkmark
        \cr

        \hline

        Power allocation     
        &\centering\checkmark
        &\centering\tikzxmark
        &\centering\checkmark 
        &\centering\checkmark  
        &\centering\checkmark  
        &\centering\checkmark     
        &\centering\checkmark 
        &\centering\checkmark
        &\centering\tikzxmark
        &\centering\checkmark 
        &\centering\checkmark
        &\centering\tikzxmark
        &\centering\checkmark
        &\centering\checkmark
         \cr
        
        \hline

         AP mode selection    
        &\centering\checkmark  
        &\centering\tikzxmark
        &\centering\tikzxmark
        &\centering\tikzxmark 
        &\centering\tikzxmark   
        &\centering\tikzxmark      
        &\centering\tikzxmark 
        &\centering\tikzxmark
        &\centering\tikzxmark 
        &\centering\tikzxmark 
        & \centering\checkmark
        &\centering\tikzxmark
        &\centering\tikzxmark
        &\centering\tikzxmark
        \cr

        \hline

       ML approach    
        &\centering\checkmark  
        &\centering\tikzxmark
        &\centering\tikzxmark
        &\centering\tikzxmark 
        &\centering\tikzxmark  
        &\centering\tikzxmark     
        &\centering\tikzxmark
        &\centering\tikzxmark
        &\centering\tikzxmark
        &\centering\tikzxmark 
        &\centering\tikzxmark
        &\centering\tikzxmark
        &\centering\tikzxmark
        &\centering\tikzxmark
        \cr               
        \hline
        
\end{tabular}
\label{Contribution}
\end{table*}

\vspace{-1em}
\subsection{Literature Review}
To date, SWIPT-enabled CF-mMIMO systems have been investigated in several studies~\cite{Wang:JIOT:2020,Demir:TWC:2021,Femenias:TCOM:2021,Xinjiang:TWC:2021,Zhang:IoT:2022,Zhang:TWC:2023,Galappaththige:WCL:2024, Mohammadi:GC:2023}.
Wang~\ettall~\cite{Wang:JIOT:2020} demonstrated the superior performance of cell-free IoT compared to the collocated mMIMO and small-cell IoT systems in terms of both SE and HE metrics. Demir~\ettall~\cite{Demir:TWC:2021} combined large-scale fading decoding with power control to enhance max-min fairness in SE and HE metrics. Femenias~\ettall~\cite{Femenias:TCOM:2021} investigated a coupled uplink (UL)/downlink (DL) optimization problem and examined the impact of phase length on maximizing the weighted DL signal-to-interference-and-noise ratio (SINR) and HE. Xia~\ettall~\cite{Xinjiang:TWC:2021} studied UL and DL CF-mMIMO SWIPT with network-assisted full-duplexing. Zhang~\ettall~\cite{Zhang:IoT:2022} proposed a max–min power control policy with the aim of achieving uniform harvested energy and DL SE in all sensors in a CF-mMIMO SWIPT IoT network. 
Galappaththige~\ettall~\cite{Galappaththige:WCL:2024} formulated a multi-objective
optimization problem to jointly maximize the harvested power and the weakest UE rate in SWIPT-enabled CF-mMIMO networks.
Zhang~\ettall~\cite{Zhang:TWC:2023} investigated a sum-rate maximization
problem in SWIPT CF-mMIMO networks with non-orthogonal multiple access.  A common approach in the studies mentioned above is to configure wireless information transfer (WIT) and wireless energy transfer (WPT) over orthogonal intervals within a time-division duplex (TDD) frame. In contrast, Mohammadi~\ettall~\cite{Mohammadi:GC:2023} proposed a joint optimization of AP mode selection and power control to enable WIT and WPT to operate simultaneously within the same interval.
\subsection{Research Gap and Summarized Contributions}
Recognizing the promising potential of RISs, their application in CF-mMIMO has been explored in the literature for various objectives, such as extending coverage, addressing blockage issues~\cite{cite:Chien:TWC:2022,Lan:IoT:2024}, and enhancing secrecy~\cite{Elhoushy:WCL:2022}. Moreover, \cite{cite:b6} is a recent study that analyzed the integration of BD-RIS into a WPT system. However, to the best of our knowledge, the potential of BD-RISs to enhance the SWIPT performance in CF-mMIMO systems under practical impairments remains unexplored. To address this gap, we present a comprehensive framework to investigate the potential gains of deploying BD-RIS in CF-mMIMO networks, and then we numerically compare the SWIPT performance between various RIS architectures. Another shortcoming of the aforementioned literature is the lack of utilization of machine learning (ML)-based algorithms to address complex non-convex optimization problems, which is tackled in this paper.

A comparison of our contributions against the state of the art is presented in Table~\ref{Contribution}.

To optimize the network resource utilization, we propose to jointly optimize the AP mode selection and power control strategy to maximize the average sum-HE of ERs, subject to per-IR SE requirements and per-AP power constraints. Specifically, using long-term channel state information (CSI), the APs are divided into information transmission APs (I-APs) and energy transmission APs (E-APs), which simultaneously serve IRs and ERs over the entire coherence interval. To manage interference between ERs and IRs, we consider the protective partial zero-forcing (PPZF) precoding design, which applies partial zero-forcing (PZF) at I-APs and protective maximum ratio transmission (PMRT) at E-APs. Compared to our recent work~\cite{Hua:WCNC:2024}, where the AP modes were randomly assigned without accounting for power control design, line-of-sight (LoS) channels and pilot contamination (PC) effects, our current analysis caters for these factors rigorously. The main contributions of our paper can be summarized as follows:

\begin{itemize}
    \item We derive closed-form expressions for the achievable SE of IRs and HE of ERs with non-linear EH circuits, considering channel estimation errors, PC, and Ricean fading. Building on these results, we formulate an optimization problem that jointly addresses the AP mode selection, power control, and design of the BD-RIS's scattering matrix. The goal is to maximize the average sum-HE of the ERs while meeting constraints on the per-ER minimum HE and per-IR minimum SE requirement.
    \item To tackle the mixed-integer non-convex optimization problem, we propose a heuristic approach for designing the unitary and symmetric non-diagonal scattering matrix. Then, for AP mode selection and power control, we utilize the successive convex approximation (SCA) with binary relaxation approach. To further reduce the optimization complexity, we introduce a deep reinforcement learning (DRL)-based method, where the optimization problem is framed as a Markov Decision Process (MDP). In this framework, the reinforcement learning (RL) agent—implemented as deep neural networks (DNNs)—dynamically adapt to the reformulated system and make informed decisions.
    \item Our numerical results demonstrate that ERs can benefit significantly from AP mode selection to enhance their EH capability. Furthermore, as the number of APs increases while the total number of antennas service in the network remains fixed (i.e., the number of antennas per-AP reduces), the energy-rate region achieved by the proposed SCA and DRL-based designs expands markedly compared to benchmark designs with random AP mode selection, equal power allocation, and a discrete Fourier transform (DFT)-based scattering matrix. Our numerical evaluations indicate that employing a BD-RIS instead of a D-RIS can significantly reduce the required massive MIMO dimensionality—specifically, by decreasing the number of serving APs by approximately a factor of six—while maintaining equivalent WPT performance.
\end{itemize}

\textit{Notation:} We use bold upper/lower case letters to denote matrices/vectors. The superscripts $(\cdot)^T$ and $(\cdot)^\dag$ stand for the transpose and the conjugate-transpose, respectively; $\text{mod}( \cdot, \cdot)$ is the modulus operation; $ \lfloor \cdot \rfloor$ denotes the truncated argument; $\mathbf{I}_N$ denotes $N\times N$ identity matrix. $\boldsymbol{1}_N$ and $\boldsymbol{0}_N$ denote the $N \times N$ matrices of all ones and all zeros, respectively. A circular symmetric complex Gaussian variable with variance $\sigma^2$ is denoted by $\mathcal{CN}(0,\sigma^2)$. Finally, $\mathbb{E}\{\cdot\}$ denotes the statistical expectation.

\section{System Model}~\label{sec:Sysmodel}
We consider a BD-RIS-assisted CF-mMIMO SWIPT system, where all APs cooperate to simultaneously serve $K$ IRs and $J$ ERs in the same frequency bands. Each AP is connected to the central processing unit (CPU) via a high-capacity fronthaul link. The APs are equipped with $L$ antennas each, while the IRs and ERs are equipped with a single antenna. All APs, IRs and ERs are half-duplex devices. For notational simplicity, we define sets $\K\triangleq \{1,\dots,K\}$ and $\mathcal{J}\triangleq\{1,\ldots,J\}$ to collect the IR and ER indices, respectively. Moreover, the set of all APs is denoted by $\MM \triangleq \{1, \dots, M\}$. To ensure SWIPT under the same frequency spectrum, we consider the selection of operation modes for APs. Therefore, APs are classified into I-APs and E-APs. The I-APs support WIT to the IRs, while the E-APs support WPT towards the ERs. To further improve the HE of the ERs, a BD-RIS is located near the ER zone to facilitate transmission, as shown in Fig.~\ref{fig:system_model}. In the BD-RIS, the reconfigurable scattering elements are interconnected and element-wise correlated. Thus, the scattering matrix $\boldsymbol{\Theta} \in\mathbb{C}^{N\times N}$ is a non-diagonal matrix of the $N$-port reconfigurable impedance
network. In general, the architecture of a scattering matrix can be either non-diagonal and fully connected (FC) or block-diagonal and group-connected (GC)~\cite{cite:Shen:BDRISoverview02:2022}. When the $N$-port reconfigurable impedance network is lossless~\cite{Li:JSAC:2023}
\begin{align}
  \boldsymbol{\Theta}^{\dag}\boldsymbol{\Theta}=\qI_N, \text{ and } \boldsymbol{\Theta} = \boldsymbol{\Theta}^T.  
\end{align} 

\subsection{Channel Model}
In the CF-mMIMO context, a TDD protocol combined with a statistical CSI-based system is preferable due to its capability of significantly reducing the pilot overhead~\cite{Hien:TDD:pilot, cite:HienNgo:cf02:2018}. In addition, traditional approaches often assumed that receivers need to know their instantaneous CSI for effective decoding~\cite{cite:zahra:statisticalCSI}. However, by virtue of \textit{channel hardening},
statistical CSI-based designs can sufficiently capture the wireless environment characteristics, including mobility effects, without requiring frequent iterative updates as in instantaneous CSI-based schemes.\footnote{According to~\cite{cite:HienNgo:cf01:2017}, at a user mobility of $100$ km/h and a carrier frequency of $f_c=2$ GHz, small-scale fading varies rapidly, with a coherence time on the order of $1$ ms, whereas large-scale fading parameters may change approximately $40$ times slower, i.e., every $40$ ms. In our framework, the precoding design is updated based on the estimated instantaneous CSI at each coherence interval, while the overall optimization framework (SCA or ML design) is based on statistical CSI, which remains constant over multiple coherence intervals and is updated less frequently.} Following such numerical insights, we adopt a widely accepted block-fading channel model, where the channel coefficients remain time-invariant and frequency-flat during each coherence interval $\tau_{c}$ and vary between the intervals. The duration of the training phase is denoted as $\tau$, while the duration of the DL WIT and WPT is $(\tau_c-\tau)$. In each coherence block, the instantaneous channel from AP $m$ to IR $k$, $\gmkiu \in \mathbb{C}^{L \times 1}$, is distributed as $\CN(\boldsymbol{0},\betamkue \qI_{L})$, where $\betamkue$ is the large-scale fading coefficient (LSFC). In addition, $\gmjue \in \mathbb{C}^{L \times 1}$ is the aggregated channel between AP $m$ and ER $j$, given by\footnote{Following~\cite{cite:Chien:TWC:2022}, during simulation, the signal blockage is modeled by randomly disabling the direct AP-receiver links with probability $1 - \tilde{p}$, where $\tilde{p}=1$ denotes the likelihood of a direct link remaining unblocked.}
\begin{figure}[t]
	\centering
	\vspace{0em}
	\includegraphics[width=80mm]{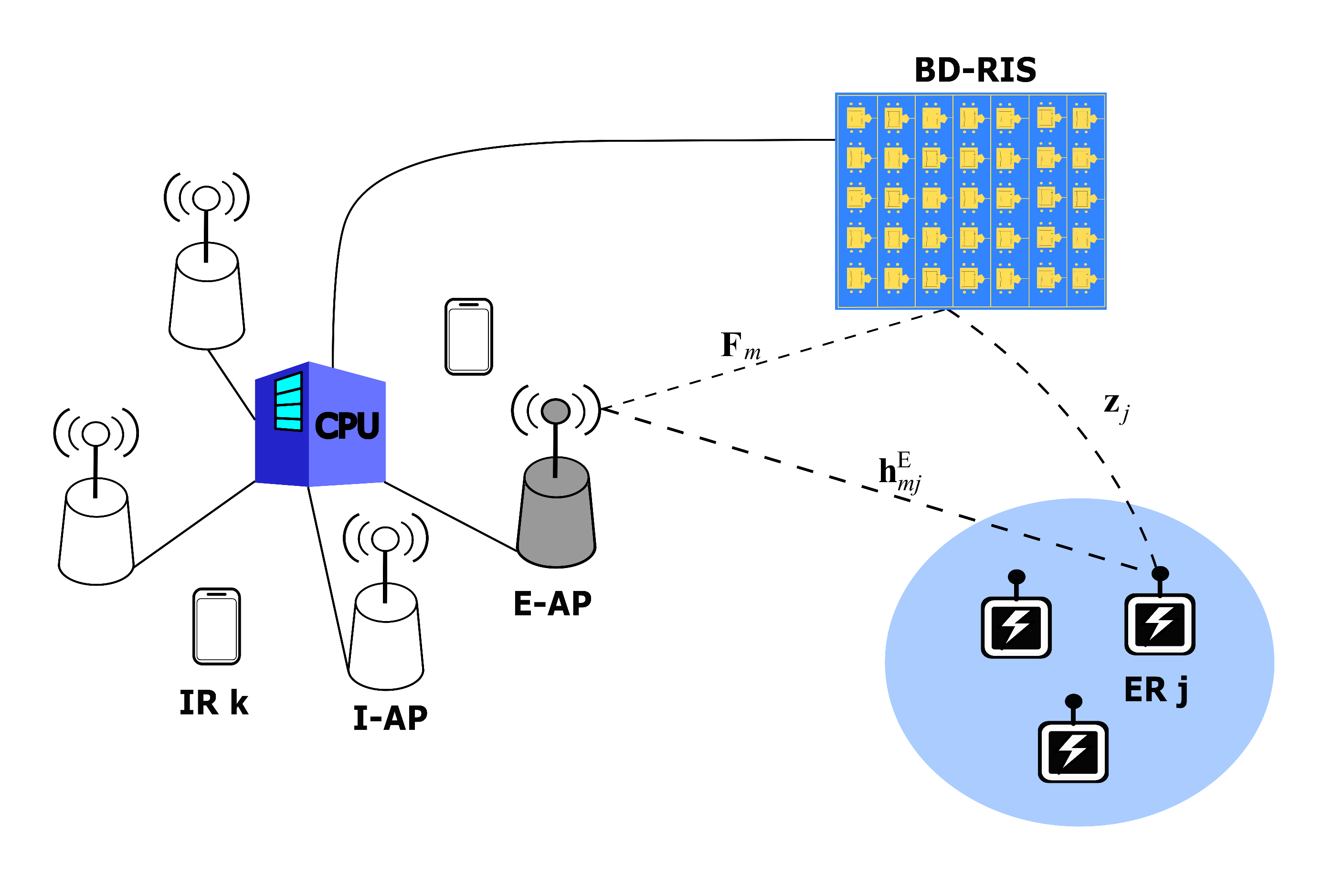}
	\vspace{-1.0em}
	\caption{\small Illustration of BD-RIS-assisted CF-mMIMO SWIPT system.}
	\vspace{0.5em}
	\label{fig:system_model}
\end{figure}
\begin{equation}~\label{eq:gmjeu}
\gmjue = \hmlue + \FmRIS \boldsymbol{\Theta}  \qz_{j},
\end{equation}
where $\hmlue\sim \CN(\boldsymbol{0}, \betamjeu \qI_{L}) \in \C^{L \times 1 }$ denotes the direct channel from  AP $m$ to ER $j$, while $\betamjeu$ is the LSFC; $\FmRIS \in \C^{L \times N}$ is the channel matrix between  AP $m$ and BD-RIS, and $\qz_{j} \in \C^{N \times 1}$ is the channel from the BD-RIS to ER $j$. A Ricean model is considered for the AP-RIS link due to the inclusion of both LoS and non line-of-sight paths (NLoS). Thus, 
\begin{equation}
    \FmRIS = \sqrt{\frac{\zeta_m}{1+\kappa}}\Big( \sqrt{\kappa}\FmRISlos + \FmRISnlos \Big),
\end{equation}
where $\zeta_m$ is the LSFC corresponding to the channel between AP $m$ to the BD-RIS; $\kappa$ is the Ricean factor; $\big[\FmRISnlos\big]_{n,l} \sim \CN(0,1)$ is the element of the NLoS component. By considering a uniform linear array at the APs and uniform squared planar array for BD-RIS,
the LoS component $\FmRISlos$ is modeled as
\begin{equation}~\label{eq:FmRISlos}
\FmRISlos = \qa_{L}(\vartheta^{\mathtt{aoa,azi}}_{m}, \vartheta^{\mathtt{aoa,ele}}_{m})
            \qa_{N}^\dag(\vartheta^{\mathtt{aod,azi}}_{m}, \vartheta^{\mathtt{aod,ele}}_{m}),
\end{equation}
where 
 \begin{align}~\label{eq:FmRISlos_aoa}
    \big[\qa_{L}(\vartheta^{\mathtt{aoa,azi}}_{m}, \vartheta^{\mathtt{aoa,ele}}_{m}) \big]_{l} 
    &\!=\!
    e^{j \frac{2\pi d_{BS}}{\lambda}(l-1) \cos{(\vartheta^{\mathtt{aoa,azi}}_{m})} \cos{(\vartheta^{\mathtt{aoa,ele}}_{m})} },
\end{align}   
where $\lambda$ is the wavelength; $d_{BS}$ is the spacing between two adjacent AP antennas; $\vartheta^{\mathtt{aoa,azi}}_{m}$ and $\vartheta^{\mathtt{aoa,ele}}_{m}$ is the elevation and azimuth angle of arrival (AoA), respectively.

Additionally, the $n$-th element of $\qa_{N}(\vartheta^{\mathtt{aod,azi}}_{m}, \vartheta^{\mathtt{aod,azi}}_{m})$ is modeled as
$\big[ \qa_{N} (\vartheta^{\mathtt{aod,azi}}_{m}, \vartheta^{\mathtt{aod,ele}}_{m}) \big]_{n} = e^{j ( \varrho_{H,n} + \varrho_{V,n} )  }$, where
\begin{subequations}~\label{eq:FmRISlos_aodii}
  \begin{align}~\label{eq:FmRISlos_aod}
&\varrho_{H,n} 
=
\frac{2 \pi d_{H} }{\lambda} [N_{H}]_{n} \sin{(\vartheta^{\mathtt{aod,azi}}_{m})} \cos{(\vartheta^{\mathtt{aod,ele}}_{m})}, \\
&\varrho_{V,n} 
=
\frac{2 \pi d_{V} }{\lambda} [N_{V}]_{n} \sin{(\vartheta^{\mathtt{aod,ele}}_{m})}.
\end{align}  
\end{subequations}
In~\eqref{eq:FmRISlos_aodii}, $N_H$ ($N_V$) denotes the horizontal (vertical) number of reconfigurable elements and $N = N_{V} \times N_{H}$, where $[N_V]_n = \lfloor (n-1)/\sqrt{N} \rfloor$ and $[N_H]_n = \mod({n-1,\sqrt{N}})$; $d_{H}$ ($d_{V}$) is the spacing between two adjacent horizontal (vertical) reconfigurable elements; $\vartheta^{\mathtt{aod,azi}}_{m}$ ($\vartheta^{\mathtt{aod,ele}}_{m}$) is the elevation (azimuth) angle of departure. Since the BD-RIS is  in the proximity of ERs, the LoS paths are dominant over the NLoS paths, and thus, it is reasonable to assume LoS channels for $\qz_{j}$, i.e.,
\begin{equation}
    \qz_{j} = \sqrt{\alpha_j} \qa_{N}(\varepsilon^{\mathtt{aoa,ele}}_{j}, \varepsilon^{\mathtt{aoa,azi}}_{j}),
\end{equation}
where $\alpha_j$ is the LSFC corresponding from the BD-RIS to  ER $j$ channel; $\varepsilon^{\mathtt{aoa,ele}}_{j}$ ($\varepsilon^{\mathtt{aoa,azi}}_{j}$) is the elevation (azimuth) of the angle of departure of the incident signal from the ER $j$ to the BD-RIS, while the $n$-th element of $\qa_{N}(\varepsilon^{\mathtt{aoa,ele}}_{j}, \varepsilon^{\mathtt{aoa,azi}}_{j})$ are similarly formulated as~\eqref{eq:FmRISlos_aod}. Therefore, the aggregated indirect channel from AP $m$ to ER $j$ can be derived as
\begin{equation}~\label{eq:aggregate_indirect_link}
\FmRIS \boldsymbol{\Theta} \qz_{j} = \sqrt{\bar{\zeta}_{mj}\kappa} \FmRISlos \boldsymbol{\Theta} \qz_{j} + \sqrt{\bar{\zeta}_{mj} } \FmRISnlos \boldsymbol{\Theta} \qz_{j},
\end{equation}
where $\bar{\zeta}_{mj} \triangleq \frac{\zeta_{m} \alpha_{j} }{1 + \kappa}$. Since the elements of $\FmRISnlos$ are independent
and identically distributed (i.i.d.)  $\CN(0,1)$, the elements of the second term in~\eqref{eq:aggregate_indirect_link} are linear combinations of independent Gaussian distribution, and hence, $\sqrt{\bar{\zeta}_{mj}} \FmRISnlos \boldsymbol{\Theta} \qz_{j} \sim \CN \big(\boldsymbol{0}, N \bar{\zeta}_{mj} \qI_{L} \big)$. Thus, $\FmRIS \boldsymbol{\Theta} \qz_{j}$ is distributed as $\CN\Big(\sqrt{\bar{\zeta}_{mj} \kappa} \FmRISlos \boldsymbol{\Theta} \qz_{j}, N \bar{\zeta}_{mj} \qI_{L} \Big)$.

For brevity, we denote $\gmjeulos \triangleq \FmRISlos \boldsymbol{\Theta} \qz_{j} \in \C^{L \times 1}$. Then, the channel from AP $m$ to ER $j$ is Gaussian distributed as $\gmjue \sim \CN \Big(\sqrt{\bar{\zeta}_{mj} \kappa} \gmjeulos, \big(\betamjeu + N \bar{\zeta}_{mj} \big) \qI_{L}  \Big)$, which is essential to facilitate subsequent derivations.
\vspace{-0.7em}
\subsection{Uplink Training for Channel Estimation}\label{phase:ULforCE}
During UL training phase, all ERs and IRs transmit their pilot sequences of length $\tau$ symbols to the APs. Practically, since $\tau < J+K$, pilot sequences are reused among the users, we denote $i_{j} (i_{k}) \in \{1, \ldots, \tau \}$ as the index of the pilot sequence $\VARPHI^{\mathtt{E}}_{i_j}$ ($\VARPHI^{\mathtt{I}}_{i_k}$) used by the ER-$j$ (IR-$k$). For a general case, we denote $\Per \subset \mathcal{J}$ ($\Pir \subset \mathcal{K}$) as the set of ERs (IRs), including $j$ ($k$), that are assigned the same pilot sequence with the ER-$j$ (IR-$k$).
In addition, we denote $\tau_{\K}$ ($K \geq \tau_{\K}$) and $\tau_{\J}$ ($J \geq \tau_{\J}$) as the number of orthogonal pilots used by $K$ IRs and $J$ ERs, respectively. Let $\boldsymbol{\Phi}^{\mathrm{I}} = [\VARPHI^{\mathrm{I}}_{1}, \ldots, \VARPHI^{\mathrm{I}}_{\tau_{\K}}] \in \C^{\tau \times \tau_{\K}}$ and $\boldsymbol{\Phi}^{\mathrm{E}} = [\VARPHI^{\mathrm{E}}_{1}, \ldots, \VARPHI^{\mathrm{E}}_{\tau_{\J}}]\in \C^{\tau \times \tau_{\J}}$ be the pilot-book matrices, such that $ (\VARPHI^{\mathrm{I}}_{j})^\dag\VARPHI^{\mathrm{I}}_{i} =0$, $ (\VARPHI^{\mathrm{E}}_{j})^\dag\VARPHI^{\mathrm{E}}_{i} =0$, and $ (\VARPHI^{\mathrm{E}}_{j})^\dag\VARPHI^{\mathrm{I}}_{i} =0$, $\forall i \neq j$. The corresponding full-rank matrices of the estimated channels are formulated as
\begin{align}
    \Ghmu &= \qY_{p,m}^{\mathtt{I}}  \boldsymbol{\Phi}^{\mathrm{I}}, \text{ and }
    \\
    \Ghms &= (\qY_{p,m}^{\mathtt{E}} -\Ex\{\qY_{p,m}^{\mathtt{E}}\}) \boldsymbol{\Phi}^{\mathrm{E}}, 
\end{align}
where $\qY_{p,m}^{\mathtt{E}}$ ($\qY_{p,m}^{\mathtt{I}}$) are the pilot signals received at E-AP (I-AP)-$m$ from ERs and IRs during the training phase. Mathematically speaking
\vspace{-0.1em}
\begin{equation}~\label{eq:receivedpilotsequence_ER}
    \qY_{p,m}^{\mathtt{i}} \!\!=\! \sqrt{\tau\rho_{u} }\Big(\!\sum\nolimits_{j\in\mathcal{J}}\!\gmjue \VARPHI^{\mathtt{E}\dag}_{i_j}\ \!+\!\sum\nolimits_{k\in\mathcal{K}}\!\gmkiu  \VARPHI^{\mathtt{I}\dag}_{i_k} \!\Big) \!+\! \qN_{p,m},
\end{equation}
where $\mathtt{i} \in \{\mathtt{I}, \mathtt{E} \}$, $\rho_{u}$ is the normalized UL signal-to-noise ratio (SNR), while $\qN_{p,m} \in \C^{L \times \tau}$ denotes the noise matrix containing i.i.d. $\CN(0,\Snn)$ RVs. To estimate the channel $\gmjue$ and $\gmkiu$ from the AP $m$ to the ER-$j$ and IR-$k$, $\qY_{p,m}^{\mathtt{E}}$ ($\qY_{p,m}^{\mathtt{I}}$) is first projected onto the $i_j$ ($i_k$)-th pilot sequence, i.e., 
\begin{align}~\label{eq:desrpeading_ER}
\qy_{mj}^{\mathtt{E}} 
&= \sqrt{\tau\rho_{u}} \Big(\gmjue +  \sum\nolimits_{j' \in \Per \setminus j} \! \gmjpue \Big)+ \Tilde{\qn}_{p,mj},
\end{align}
and
\begin{align}~\label{eq:desrpeading_IR}
\qy_{mk}^{\mathtt{I}} 
&= \sqrt{\tau\rho_{u}} \Big(\gmkiu +  \sum\nolimits_{k' \in \Pir \setminus k} \! \gmkpue \Big) + \Tilde{\qn}_{p,mk},
\end{align}
where $\Tilde{\qn}_{p,mj} \triangleq \qN_{p,m} \VARPHI^{\mathtt{E}}_{i_j} \sim \CN(\boldsymbol{0},\Snn \qI_{L})$ and $\Tilde{\qn}_{p,mk} \triangleq \qN_{p,m} \VARPHI^{\mathtt{I}}_{i_k} \sim \CN(\boldsymbol{0},\Snn \qI_{L})$. Invoking the analytical process as in \cite[Eq. (15.64)]{cite:book:Kay:1993}, the minimum mean-squared error (MMSE) estimate of $\gmjue$ can be obtained as
\begin{align}
    \hgmjue &\!= \!\sqrt{\bar{\zeta}_{mj} \kappa} \gmjeulos + c_{mj}^{\mathrm{E}} 
    \Big(\! \qy_{mj}^{\mathrm{E}} \!-\!  \sqrt{\tau \rho_{u}}
        \sum\nolimits_{j' \in \Per } 
        \!\! { \sqrt{ \bar{\zeta}_{mj'} \kappa} \gmjpeulos }
        \!\Big),
\end{align}
where
\vspace{-0.8em}
\begin{equation}
    c_{mj}^{\mathrm{E}} \!\!\triangleq\!\! \frac{
    \sqrt{\tau \rho_{u}} \big( \betamjeu + N \bar{\zeta}_{mj} )
    }{
    \tau \rho_{u} \sum_{j' \in \Per}{ \big( \betamjpeu + N \bar{\zeta}_{mj'} )   } 
    \!+\! \Snn
    }.
\end{equation}
Thus, $\hgmjue \sim \CN \big(\sqrt{\bar{\zeta}_{mj} \kappa} \gmjeulos, \gameumj \qI_{L}  \big)$, where
\begin{equation}
    \gameumj \triangleq \Ex \Big\{ \big\Vert [\hgmjue]_{l} \big\Vert^{2} \Big\} = \sqrt{\tau \rho_{u}}(\betamjeu + N\bar{\zeta}_{mj})c_{mj}^{\mathrm{E}}.
\end{equation}
Moreover, $\gtilmjeu \triangleq \gmjue - \hgmjue$ is the estimation error of $\hgmjue$, which is distributed as $\CN \big(\boldsymbol{0}, ( \betamjeu + N\bar{\zeta}_{mj} - \gameumj ) \qI_{L} \big)$. 

By applying a similar process, the MMSE estimate of $\gmkiu$ is obtained as
\begin{align}
    \hgmkue = c_{mk}^{\mathrm{I}} \Big(\qy_{mk}^{\mathrm{I}} -  \sqrt{\tau \rho_{u}}
        \sum\nolimits_{j' \in \Per } 
        \!\! { \sqrt{ \bar{\zeta}_{mj'} \kappa} \gmjpeulos }\Big),
\end{align}
where
\vspace{-0.7em}
\begin{equation}
    c_{mk}^{\mathrm{I}} \!\!\triangleq\!\! \frac{
    \sqrt{\tau \rho_{u}} \betamkue  
    }{
    \tau \rho_{u}  \sum_{k' \in \Pir}{ \betamkpue  }
    +  \Snn 
    }.
\end{equation}
Therefore, the estimated channel $\hgmkue$ is distributed as $\CN(\boldsymbol{0},\gamuemk \qI_{L})$ where $\gamuemk$ can be computed as
\begin{equation}~\label{eq:cov_hgmkue}
\gamuemk \triangleq \Ex \big\{ \big\Vert [\hgmkue]_{l} \big\Vert^{2} \big\} = \sqrt{\tau \rho_{u}}\betamkue c_{mk}^{\mathrm{I}}.
\end{equation}
Hence, the estimation error of $\gmkiu$ is denoted as $\gtilmkue$, which is distributed as $\gtilmkue \sim \CN \big(\boldsymbol{0}, (\betamkue - \gamuemk) \qI_L \big)$.

Given $c_{mj}^{\mathrm{E}}$ and $c_{mk}^{\mathrm{I}}$, the corresponding full-rank estimated channels between AP $m$ and ER $j$, as well as between AP $m$ and IR $k$, are expressed as 
\begin{align}
    \hgmjue = \sqrt{\bar{\zeta}_{mj} \kappa} \gmjeulos \nonumber + c_{mj}^{\mathrm{E}}\Ghms \qe^{\mathrm{E}}_{\tau_{\J}} \text{, and, }
    \hgmkue = c_{mk}^{\mathrm{I}} \Ghmu \qe^{\mathrm{I}}_{\tau_{\K}},
\end{align}
where $\qe^{\mathrm{E}}_{\tau_{\J}}$ and $\qe^{\mathrm{I}}_{\tau_{\K}}$ denote the $j$-th and $k$-th column of $\qI_{\tau_{\J}}$ and $\qI_{\tau_{\K}}$, respectively.

\begin{Remark}~\label{remark:PC}
When ER $j$ and $j'$ ($j' \neq j$) share the same pilot sequences, the respective NLoS components $\gtilhatmjpeu$ and $\gtilhatmjeu$ of the channel estimates are linearly dependent as 
\begin{equation}
    \gtilhatmjpeu  \!=\! \upsilon_{m,j,j'} \gtilhatmjeu,
\end{equation}
where $\upsilon_{m,j,j'} \triangleq (\betamjpeu \!+\! N \bar{\zeta}_{mj'}) / (\betamjeu \!+\! N \bar{\zeta}_{mj} )$. In addition, $\gameumjp  \!=\! \upsilon_{m,j,j'}^{2} \gameumj$.
\end{Remark}

\subsection{Downlink Wireless Information and Power Transmission}
We define the binary variable $a_m$ as the operation indicator, where $a_m = 1$ if the $m$-th AP operates as an I-AP, and $a_m = 0$ otherwise. Moreover, $\xik$ ($\xej$) denotes the data (energy) symbol transmitted to the IR $k$ and satisfy $\Ex\{ \vert \xik \vert^2 \} = \Ex\{ \vert \xej \vert^2 \} =1$. Hence, the transmitted signal from the $m$-th AP can be expressed as
\vspace{-0.1em}
\begin{align}~\label{eq:x_m}
\qx_{m}
&= \sum\nolimits_{k\in\K}\sqrt{a_m\rho_{d}\etamkI} \wimk \xik \nonumber \\
&+  \sum\nolimits_{j\in\J} \sqrt{(1-a_m)\rho_{d}\etamjE}\wemj \xej,
\end{align}
where $\rho_{d} = \tilde{\rho}_{d}/\Snn$ is the normalized DL SNR; $\wimk \in \C^{L\times 1}$ ($\wemj\in \C^{L\times 1}$) represents the precoding vector for IR $k$ (ER $j$), subject to the constraint $\Ex\big\{\big\Vert\wimk\big\Vert^2\big\}=\Ex\big\{\big\Vert\wemj\big\Vert^2\big\}=1$. Furthermore, $\etamkI$ ($\etamjE$) denote the power control coefficients associated with AP $m$ and IR $k$ (ER $j$) chosen to satisfy the power constraint at each AP, i.e., $\Ex\{\Vert \qx_m\Vert\}\leq \rho_{d}$, which is equivalent to
\begin{align}\label{eq:APpower}
& a_m\!\sum\nolimits_{k\in\K} \!\etamkI\! +\! (1\!-\!a_m) \sum\nolimits_{j\in\J}\!\etamjE\!\leq\! 1.
\end{align}
With the transmitted signal $\qx_m$, the $k$-th IR and the $j$-th ER then, respectively, receive
\vspace{-0.1em}
\begin{subequations}
 \begin{align}
    \yik 
    &=  
    \sum\nolimits_{k'\in \K}\sum\nolimits_{m\in\M} \sqrt{a_m \rho_d\etamkpI} (\gmkiu)^\dag \wimkp \xikp \nonumber \\
    &\hspace{-1em}+ \! \sum\nolimits_{j\in\J}\sum\nolimits_{m\in\M}\!\! \sqrt{(\!1\!-\!a_m\!) \rho_d\etamjE} (\gmkiu)^\dag\! \wemj \xej \! +\! n_{d,k},~\label{eq:yik} \\
    \yej
    &=\!  \sum\nolimits_{j'\in\J}\!\sum\nolimits_{m\in\M} \!\!\sqrt{(1\!-\!a_m) \rho_d\etamjpE} (\gmjue)^\dag \wemjp \xejp   \nonumber \\
    &\hspace{-1.7em}+\!\!  \sum\nolimits_{k\in\K}\!\sum\nolimits_{m\in\M}\!\! \sqrt{\!a_m \rho_d\etamkI} (\gmjue)^\dag \wimk \xik \!\! +\! n_{d,j}~\label{eq:yel}
\end{align}
\end{subequations}
where $n_{d,k},n_{d,j} \!\sim \! \CN(0,1)$ are the additive noise terms.
\vspace{-0.4em}
\subsection{Protective Partial Zero Forcing}
We apply PPZF precoding, in which local PZF precoding is considered at the I-APs and PMRT is applied at the E-APs. The principle behind this design is to guarantee full protection for the IRs against energy signals intended for ERs. In addition, power transfer with MRT proves to be optimal for WPT systems, especially with a large number of antennas~\cite{cite:MRT_for_HE:Almradi}. Hence, we note that PPZF is efficient in suppressing non-coherent interference, making it nearly optimal for SWIPT particularly in the mMIMO context~\cite{cite:pzf_pmrt_2020, Mohammadi:GC:2023}. The PZF and PMRT precoders at the AP $m$ for IR $k$ and ER $j$ can be designed as
\vspace{-0.4em}
\begin{subequations}
 \begin{align}
    \wimk^{\PZF} &= \sqrt{\alpha_{mk}^{\PZF}}{ \Ghmu \Big(\big(\Ghmu\big)^H \Ghmu\Big)^{-1} \qe_k^I},~\label{eq:wipzf}
    \\
    \wemj^{\PMRT} &= \sqrt{\alpha_{mj}^{\PMRT}}{\qB_m\Ghms\qe_{j}^{\mathrm{E}}},~\label{eq:wemrt}
\end{align}   
\end{subequations}
with $\qe_{k}^{\mathrm{I}}$ and $\qe_{j}^{\mathrm{E}}$ being the $k$-th column of $\qI_{K}$ and the $j$-th column of $\qI_{J}$, respectively; $\qB_m$ denotes the projection matrix onto the orthogonal complement of $\Ghmu$ so that $\big(\ghmkue\big)^H \qB_m =\boldsymbol{0}$. Thus, $\qB_m$ can be computed as
\begin{align}
  \qB_m  = \qI_{L}  - \Ghmu \Big( \big(\Ghmu\big)^H \Ghmu\Big)^{-1}  \big(\Ghmu\big)^H.
\end{align}
Moreover, according to~\cite{Mohammadi:TC:2024}, we have 
\begin{subequations}
\begin{align}
    \alpha_{mk}^{\PZF} &= {(L-K)\gamuemk}~\label{eq:alpha_wipzf},
    \\
    \alpha_{mj}^{\PMRT}&= \Big( \frac{L-K}{L}( L\gameumj + \bar{\zeta}_m\kappa (\gmjeulos)^\dag \gmjeulos ) \Big)^{-1}.~\label{eq:alpha_wemrt}
\end{align}
\end{subequations}

\section{Performance Analysis and Resource Allocation Formulation}
\subsection{Downlink Achievable SE and Average HE}
Upon receiving \eqref{eq:yik}, IR $k$ detects its desired symbol $\xik$. In the absence of a DL training phase, the IRs avail of channel statistics, specifically the average effective channel gain, to detect their desired symbols. It is worth noting that the long-term dynamic pattern, probability of the LoS components and interference are inherently captured through the channel statistics. Therefore, for the DL SE analysis at the IRs, we apply the hardening bound~\cite{cite:HienNgo:cf01:2017}. Consequently, the received signal at the $j$-th IR can be expressed as
\vspace{-0.1em}
\begin{align}~\label{eq:yi:hardening}
    \yik &=  \mathrm{DS}_k  \xik +
    \mathrm{BU}_k \xik 
         +\sum\nolimits_{k'\in\K \setminus k}
     \mathrm{IUI}_{kk'}
     \xikp\nonumber\\
    &\hspace{1em}
    + \sum\nolimits_{j\in\J}
     \mathrm{EUI}_{kj}\xej + n_k,~\forall k\in\K,
\end{align}
where $\mathrm{DS}_k$, $\mathrm{BU}_k$, $\mathrm{IUI}_{kk'}$, and $\mathrm{EUI}_{kj}$ represent the desired signal, the beamforming gain uncertainty, the interference cause by the $k'$-th IR, and the interference caused by the $j$-th ER, respectively, given by
\begin{subequations}
  \begin{align}~\label{eq:yi:components}
    \mathrm{DS}_k &\triangleq \sum\nolimits_{m\in\MM} \sqrt{a_m \rho_d \etamkI} \Ex \big\{(\gmkiu)^\dag \wimk^{\PZF} \big\},  
    \\
    \mathrm{BU}_k &\triangleq \sum\nolimits_{m\in\MM} \sqrt{a_m \rho_d \etamkI} \Bigl( (\gmkiu)^\dag \wimk^{\PZF} \nonumber \\
    &\hspace{1em}- \Ex \big\{(\gmkiu)^\dag \wimk^{\PZF} \big\} \Bigl)~\label{eq:component_BUk}, 
    \\
    \mathrm{IUI}_{kk'} &\triangleq \sum\nolimits_{m\in\MM} \sqrt{a_m \rho_d\etamkpI} (\gmkiu)^\dag \wimkp^{\PZF},~\label{eq:IU_interference}
    \\
    \mathrm{EUI}_{kj} &\triangleq \sum\nolimits_{m\in\MM} \sqrt{(1-a_m) \rho_d\etamjE} (\gmkiu)^\dag \wemj^{\PMRT}.  
    \end{align}  
\end{subequations}

The DL SE in [bit/s/Hz] for IR $k$ can be obtained as
\vspace{0.2em}
\begin{align}~\label{eq:SEk:Ex}
    \mathrm{SE}_k
      &=
      \Big(1\!- \!\frac{\tau}{\tau_c}\Big)
      \log_2
      \left(
       1\! + \SINRk\big(\boldsymbol{\Theta}, \qa,  \ETAI, \ETAE\big)
     \right),
\end{align}
where $\qa$ is the indicator vector, whose entries are $a_m$; $\ETAI = [\eta_{m1}^{\mathtt{I}}, \ldots, \eta_{mK}^{\mathtt{I}}]$; $ \ETAE = [\eta_{m1}^{\mathtt{E}}, \ldots, \eta_{mJ}^{\mathtt{E}}]$; and $\SINRk$ is the effective SINR at IR $k$,  given by: $\SINRk\big(\boldsymbol{\Theta}, \qa,  \ETAI, \ETAE\big)=$ 
\vspace{-0.6em}
\begin{align}~\label{eq:SINE:general}
       &\!\frac{
                 \big\vert  \mathrm{DS}_k  \big\vert^2
                 }
                 {  
                 \Ex\big\{\! \big\vert  \mathrm{BU}_k  \big\vert^2\!\big\} \!+
                  \!\!
                 \sum_{k'\in\K\setminus k}\!
                  \Ex\big\{\! \big\vert \mathrm{IUI}_{kk'} \big\vert^2\!\big\}
                  \! + \!
                   \!
                  \sum_{j\in\J}\!
                 \Ex \big\{ \!\big\vert  \mathrm{EUI}_{kj} \big\vert^2\!\big\}
                   \!+\!  1}. 
\end{align}

To characterize the HE at ER $j$, a non-linear EH model with the sigmoidal function is used~\cite{Boshkovska:CLET:2015}. Therefore, the total HE at ER $j\in\J$ is given by
 \vspace{-0.3em}
  \begin{align}~\label{eq:NLEH}
  \Phi^{\mathrm{NL}}_{j}\big(\boldsymbol{\Theta}, \qa,  \ETAE,\ETAI\big) = \frac{\Lambda\big(\mathrm{E}_{j}(\boldsymbol{\Theta}, \qa,  \ETAE, \ETAI)\big) - \phi \Omega }{1-\Omega},
 \end{align}
 where $\phi$ is the maximum output DC power, $\Omega=\frac{1}{1 + \exp(\xi \chi)}$ is a constant to guarantee a zero input/output response, while $\xi$ and $ \chi$ are constant related parameters that depend on the circuit. Moreover, $\Lambda\big(\mathrm{E}_{j}(\boldsymbol{\Theta}, \qa,  \ETAE, \ETAI)\big)$  is the traditional logistic function, given by 
 \vspace{-0.2em}
  \begin{align}~\label{eq:PsiEl}
     \Lambda\big(\mathrm{E}_{j}(\boldsymbol{\Theta}, \qa,  \ETAE, \ETAI)\big) &\!=\!\!\frac{\phi}{1 \!+ \!\exp\big(\!-\xi\big(\mathrm{E}_{j}(\boldsymbol{\Theta}, \qa, \ETAE,\ETAI)\!-\! \chi\big)\!\big)}.
 \end{align}
Note that $\mathrm{E}_{j}(\qa, \ETAE,\ETAI)$ denotes the received RF energy at ER $j$, $\forall j\in\J$. We consider the average of the harvested energy as the performance metric of the WPT operation, given by
  \begin{align}~\label{eq:NLEH:av}
  \Ex\big\{\Phi^{\mathrm{NL}}_{j}\big(\boldsymbol{\Theta}, \qa,  \ETAE,\ETAI\big)\big\} &= \frac{\Ex\big\{\Lambda\big(\mathrm{E}_{j}(\boldsymbol{\Theta}, \qa,  \ETAE, \ETAI)\big)\big\} \!-\! \phi \Omega }{1-\Omega}\nonumber\\
  &\hspace{-1em} \approx
  \frac{\Lambda\big(\mathcal{Q}_{j}(\boldsymbol{\Theta}, \qa,  \ETAE, \ETAI)\big) - \phi \Omega }{1-\Omega},
 \end{align}
where we have used the following approximation~\cite{Mohammadi:TC:2024}
\vspace{-0.1em}
\begin{align}\label{eq:Jensen}
 \Ex\left\{\Lambda\left(\mathrm{E}_{j}\big(\boldsymbol{\Theta}, \qa,  \ETAI, \ETAE\big)\right)\right\}&
 \approx \Lambda\left(Q_{j}(\boldsymbol{\Theta}, \qa,  \ETAI, \ETAE\big)\big\}\right),
 \end{align}
wherein $Q_{j}(\qa,  \ETAI, \ETAE\big)\triangleq\Ex\left\{\mathrm{E}_{j}\big(\qa, \ETAI, \ETAE\big)\right\}$ is the average of the received RF energy at the ER $j\in\J$, given by~\eqref{eq:El_average} at the top of the next page.
\begin{figure*}
\begin{align}~\label{eq:El_average}
     Q_{j}(\boldsymbol{\Theta}, \qa,  \ETAI, \ETAE\big)
    & =(\tau_c-\tau)\Snn\Big( {\rho}_d\!\sum\nolimits_{m\in\MM}\!\!
   {(1\!-\!a_m)\etamjE} \Ex\Big\{\!\big\vert\big(\gmjue\big)^\dag\wemj^{\PMRT}\big\vert^2\!\Big\} 
     +{\rho}_d\!\!\sum\nolimits_{j' \in\J\setminus j}\!\sum\nolimits_{m\in\MM}\!\!
   {(1\!-\!a_m)\etamjpE} 
   \nonumber\\
   &\times\Ex\Big\{\!\big\vert\big(\gmjue\big)^\dag\wemjp^{\PMRT}\big\vert^2\!\Big\} 
   +\!{\rho}_d\!
   \sum\nolimits_{k\in\K}\!\sum\nolimits_{m\in\MM}\!\!\!
   {a_m\etamkI}\Ex\Big\{\!\big\vert\!\big(\gmjue\big)^{\!\dag}\wimk^{\PZF}\big\vert^2\!\Big\} \!+\! 1 \!\Big).
\end{align}   
\hrulefill
	\vspace{-1.5em}
\end{figure*}

We now provide closed-form expressions for the SE and average HE under the PPZF precoding scheme. 

\begin{proposition}~\label{Theorem:SE:PPZF}
With PPZF precoding, the achievable SE of the $k$-th IR is given by~\eqref{eq:SEk:Ex}, where the effective SINR is 
\begin{align}~\label{eq:SINE:PPZF}
    &\SINRk^{\PZF}\big(\boldsymbol{\Theta}, \qa,  \ETAI, \ETAE\big) = \\
    &\!\frac{
            (L-K)\Big(\sum\nolimits_{m\in\M}\sqrt{ a_m\etamkI \gamuemk}  \Big)^2
            }
            { 
            \mathsf{PC}_k 
            \!+\!
            \sum_{k'\in \K} \sum\nolimits_{m \in \M}  a_m  \etamkpI \big(\betamkue - \gamuemk \big)
            + \mathsf{EUI}_k
            +  1/\rho_{d}
            },\nonumber
\end{align}
where $\mathsf{PC}_k=\sum_{k'\in \Pir \setminus k}  \Big( \sum\nolimits_{m \in \M} \sqrt{ a_m \etamkpI (L-K) \gamuemk} \Big)^{2}$ and $ \mathsf{EUI}_k=         \sum_{j\in\mathcal{J}} \sum_{m\in\MM} (1 - a_m)\etamjE(\betamkue - \gamuemk)$ denote the interference caused by PC and inter-IR interference, respectively.
\end{proposition}
\begin{proof}
    See Appendix \ref{appendix:B}.
\end{proof}

\begin{proposition}~\label{Theorem:RF:PPZF}
With PPZF precoding, the average HE at  ER $j\in\J$ is given by~\eqref{eq:NLEH:av}, where $Q_{j}(\qa,  \ETAI, \ETAE\big)$ is given by~\eqref{eq:El_average:PPZF}, at the top of the next page,
\begin{figure*}
\begin{align}~\label{eq:El_average:PPZF}
    Q_{j}(\boldsymbol{\Theta}, \qa,  \ETAI, \ETAE\big) =
    (\tau_c - \tau)\Snn
    &\Bigg[
    \rho_{d} \sum\nolimits_{m\in\M}
    (1-a_m) \frac{\etamjE}{\alpha_{mj}^{\PMRT}} {\big(\Psi_1 + \Psi_2\big)} 
    +
    {\rho_{d}}\sum\nolimits_{j'\in \Per \setminus j}  \sum\nolimits_{m\in\M}(1 - a_m)\frac{\etamjpE}{\alpha_{mj'}^{\PMRT}} 
   {\big(\Phi_1 + \Phi_2\big)}
    \nonumber\\
    &\hspace{-12em} +
    \rho_{d}
    \!\! \sum\nolimits_{j' \notin \Per \setminus j}
      \sum\nolimits_{m\in\M}(1 - a_m)
      \frac{\etamjpE}{\alpha_{mj'}^{\PMRT}} 
    \frac{(L\!-\!K)^2}{L^2 } 
    \trace
    \bigg( \!\!
    \Big( \!
    \gameumjp \qI_{L} \!+\! \bar{\zeta}_{mj'} \kappa \gmjpeulos \! \big(\gmjpeulos \!\big)^\dag  
    \Big) \!
    \Big(
    \big(\betamjeu \!+\! N \bar{\zeta}_{mj} \big) \qI_{L}
    \!+\!
    \bar{\zeta}_{mj} \kappa \gmjeulos \big(\gmjeulos \!\big)^\dag
    \Big) \!\!
    \bigg)
    \nonumber\\
    &\hspace{-12em} +
    \rho_{d}
    \sum\nolimits_{k\in\mathcal{K}}\sum\nolimits_{m\in\M}a_m \etamkI
    \Big( \betamjeu \!+\! N \bar{\zeta}_{mj} \!+\! \frac{ \kappa}{L}\bar{\zeta}_{mj} \big\Vert \gmjeulos \big\Vert^{2}  \Big)
    + 1
    \Bigg].
\end{align}
\hrulefill
	\vspace{-1.5em}
\end{figure*}
with\footnote{To derive $\Psi_1$, $\Psi_2$, $\Phi_1$, and $\Phi_2$, we exploit the approximation $\Ex\{x^2\} \approx (\Ex\{x\})^2$, which is provided in details in Appendix C.} 
\begin{align}~\label{eq:Psi_1}
    \Psi_1 
    &\!\approx\!
     \frac{L\!-\!K}{L}\kappa \bar{\zeta}_{mj} \big\Vert \gmjeulos \big\Vert^2 
    \bigg(\!\frac{L\!-\!K}{L}\kappa \bar{\zeta}_{mj} 
    \big\Vert \gmjeulos \big\Vert^2
    \!+\!(L\!-\!K\!+\!1)
    \nonumber\\
    &\hspace{-1em}
    \times \!(L\!-\!K)\big(\gameumj\big)^{\!2}
    \!+\!
    2\big(\!N\bar{\zeta}_{mj}\!+\!\betamjeu \!+\!(L\!-\!K\!-1) \gameumj \!\big),
\end{align}   
\begin{align}
    \Psi_2 
    &\!=\!
    \frac{L\! -\!K}{L}
    \big(\betamjeu \!\!+\! N\bar{\zeta}_{mj}\! - \!\gameumj \big) \Big( L\gameumj \!+\! \bar{\zeta}_{mj} \kappa \big\Vert \gmjeulos \big\Vert^2 \Big), 
\end{align}   
\begin{align}
     \Phi_1 
    &\approx
   \frac{L-K}{L}
    \bigg(  \frac{L-K}{L}\kappa^{2} \bar{\zeta}_{m,j} \bar{\zeta}_{m,j'} \big\Vert \gmjeulos \big\Vert^{2} \big\Vert \gmjpeulos \big\Vert^{2}
    \nonumber\\
    &\hspace{1em}  +
    \kappa \bar{\zeta}_{mj}(N\bar{\zeta}_{mj'} + \betamjpeu - \gameumjp)\big\Vert \gmjeulos \big\Vert^{2}
    \nonumber\\
    &\hspace{1em}
    \kappa\bar{\zeta}_{mj'}(N\bar{\zeta}_{mj} + \betamjeu - \gameumj) \big\Vert \gmjpeulos \big\Vert^{2}
    \nonumber\\
    & \hspace{1em}+
     \gameumj \upsilon_{m,j,j'} 
    \Big(L(L-K\!+1) \gameumj \upsilon_{m,j,j'}  \nonumber\\
    &\hspace{1em}
    +(L\!-\!K)
    \sqrt{\kappa^{2} \bar{\zeta}_{m,j} \bar{\zeta}_{m,j'}}
    \Big( \gmjpeulos \!+\! \gmjeulos  \Big)^\dag\gmjeulos\Big)\bigg), 
\end{align}   
\begin{align}
    \Phi_2 
    &\!=\!
    \frac{L \! - \! K}{L}
            \big(\betamjeu \! + \! N\bar{\zeta}_{mj} \!-\! \gameumj \big) \Big( L \gameumjp \! +\! \bar{\zeta}_{mj'} \kappa \big\Vert \gmjpeulos \big\Vert^2 \Big)\Big). 
\end{align}    
\end{proposition}
\begin{proof}
    See Appendix \ref{appendix:C}.
\end{proof}
\begin{Remark}~\label{remark:PC_NLHE}
    When all ERs share the same pilot sequence, the HE in~\eqref{eq:El_average:PPZF} is reduced to~\eqref{eq:El_average:PPZF:FPC} at the top of next page. 
\end{Remark}
\begin{figure*}
\begin{align}~\label{eq:El_average:PPZF:FPC}
   \tilde{Q}_{j}(\boldsymbol{\Theta}, \qa,  \ETAI, \ETAE\big)& =
    (\tau_c - \tau)\Snn
    \Bigg[
    \rho_{d} \sum\nolimits_{m\in\M} (1-a_m)
    \frac{\etamjE}{\alpha_{mj}^{\PMRT}} (\Psi_1 + \Psi_2)
    +
    {\rho_{d}}\sum\nolimits_{j'\in \Per \setminus j}  \sum\nolimits_{m\in\M}(1 - a_m) \nonumber\\
     &\times
    \frac{\etamjpE}{\alpha_{mj'}^{\PMRT}}
    (\Phi_1 + \Phi_2)
     +
     \rho_{d}
     \sum\nolimits_{k\in\mathcal{K}}\sum\nolimits_{m\in\M}a_m \etamkI
     \Big( \betamjeu \!+\! N \bar{\zeta}_{mj} \!+\! \bar{\zeta}_{mj} \frac{\kappa}{L} \big\Vert \gmjeulos \big\Vert^{2}  \Big)
     + 1
   \Bigg].
\end{align}
\hrulefill
\vspace{-1.2em}
\end{figure*}
\vspace{-0.7em}
\begin{Remark}~\label{remark:PC_NLHE_2}
    As the LoS component $\Vert \gmjeulos \Vert^{2}$ greatly contributes to the HE, $\Psi_1$, $\Psi_2$, $\Phi_1$ and $\Phi_2$ dominantly contribute to $Q_j$ and $\bar{Q}_j$ compared to other terms. This observation will be invoked for our scattering matrix design at the BD-RIS.
\end{Remark}
\begin{Remark}~\label{remark:closedform}
    The closed-form expressions for the SE and HE under PPZF precoding capture comprehensively the impact of channel estimation errors. Hence, the optimization frameworks, proposed in Section~\ref{sec:SCAOptimization}, are carried out by taking into account the channel estimation error and other statistical parameters.
\end{Remark}
\vspace{-1.2em}
\subsection{Resource Allocation Formulation}~\label{ref:ProblemFormulation}
In this subsection, we formulate an optimization problem to maximize the average sum-HE subject to practical key constraints: the minimum harvested energy required per ER, $\Gamma_{j}$, $\forall j \in \J$, the minimum quality of service (QoS) requirements per IR, $\SEth$, $\forall k \in \K$, and the power budget for each AP.\footnote{To incorporate fairness among IRs or ERs, max-min SE and max-min HE optimization problems can be alternatively formulated. Through appropriate algebraic reformulations, these problems can be transformed into a structure similar to $(\mathcal{P}1)$. Consequently, the same solution approach can be applied to effectively address them.} Our objective is to determine the optimal AP mode selection variables, $\qa$, the power control coefficients, $\ETAI, \ETAE$, and the BD-RIS scattering matrix $\boldsymbol{\Theta}$ for the considered ergodic CF-mMIMO SWIPT system. Define $\mathcal{T}_k = 2^{\SEth}-1$. The optimization problem can be mathematically expressed as
\vspace{-0.1em}
\begin{subequations}~\label{eq:ProblemFormulationorigin}
    \begin{align}
        (\mathcal{P}1): & \quad \max_{\qa, \ETAI, \ETAE, \boldsymbol{\Theta}} \quad \sum\nolimits_{j\in\J}{\Ex\big\{\Phi^{\mathrm{NL}}_{j}\big(\boldsymbol{\Theta},\qa,  \ETAE,\ETAI\big)\big\}}~\label{obj:ObjectiveFunction1}
        \\
        \mathrm{s.t.} 
        &\hspace{1em} \Ex\big\{\Phi^{\mathrm{NL}}_{j}\big(\boldsymbol{\Theta}, \qa,  \ETAE,\ETAI\big)\big\} \geq \Gamma_{j}, \forall j \in \J,~\label{ct:HE_threshold}
        \\
        &\hspace{1em} \text{SINR}_k(\boldsymbol{\Theta}, \qa, \ETAI, \ETAE) \geq \mathcal{T}_k, \forall k \in \K,~\label{ct:SINRThreshol}
        \\
        &\hspace{1em}a_m\!\sum\nolimits_{k\in\K} \!\etamkI\! +\! (1\!-\!a_m) \sum\nolimits_{j\in\J}\!\etamjE\!\leq\! 1, \forall m \!\in\! \M,~\label{ct:etamkI}
        \\
        &\hspace{1em} \boldsymbol{\Theta}^\dag \boldsymbol{\Theta} = \qI_{N}, \boldsymbol{\Theta} = \boldsymbol{\Theta}^{T},~\label{ct:BDRISConstraint}
        \\
        &\hspace{1em} a_m \in \{0,1\},\forall m \in \M.~\label{ct:a_m}
    \end{align}
\end{subequations}

\vspace{-0.7em}
\section{Proposed Solution For Average SUM-HE Maximization Problem}~\label{sec:SCAOptimization}
In this section, we develop an optimization framework to solve the sum-HE optimization problem in~\eqref{eq:ProblemFormulationorigin}. Note that both the objective function and the constraints of~\eqref{eq:ProblemFormulationorigin} are complicated functions of  the scattering matrix $\boldsymbol{\Theta}$. This issue makes the sum-HE problem technically challenging, hence it is difficult to find its optimal solution. This challenging problem is effectively transformed into tractable forms and  efficient algorithms are proposed for solving them. To this end, we propose a heuristic algorithm (\HEU) for scattering matrix design, which simplifies the computation, while providing a significant sum-HE performance gain. Furthermore, we develop a framework for joint AP mode selection and power control, utilizing SCA and ML techniques.
\vspace{-1em}
\subsection{Joint AP Mode Selection and Power Control (\textbf{JAP-PA})}~\label{sec:jap-pa}
Before proceeding, by invoking~\eqref{eq:El_average} and after some manipulation,  we can  replace constraint \eqref{ct:HE_threshold} by
\begin{equation}~\label{eq:HEinsight}
    Q_j \geq \Xi(\tilde{\Gamma}_{j}),~\forall j \in \J,
\end{equation}
where $\tilde{\Gamma}_{j} = (1-\Omega) \Gamma_{j} + \phi\Omega$ and  $\Xi(\tilde{\Gamma}_{j}) = \chi - \frac{1}{\xi}\Big( \frac{\phi -\tilde{\Gamma}_{j}}{\tilde{\Gamma}_{j}} \Big)$ is the inverse function of \eqref{eq:PsiEl}.  

For a given scattering matrix $\boldsymbol{\Theta}$, the sum-HE problem~\eqref{eq:ProblemFormulationorigin}  can be formulated as
\vspace{-0.2em}
\begin{subequations}~\label{eq:ProblemFormulationNL1}
    \begin{align}
        (\mathcal{P}2): & \quad \max_{\qa, \ETAI, \ETAE} \quad \sum\nolimits_{j\in\J}{\Ex\big\{\Phi^{\mathrm{NL}}_{j}\big(\qa,  \ETAE,\ETAI\big)\big\}}
        ~\label{obj:ObjectiveFunctionNL1}
        \\
        \mathrm{s.t.} 
        &\quad 
        Q_j \geq \Xi(\tilde{\Gamma}_{j}),~\forall j \in \J,~\label{ct:HE_threshold2}
        \\
        &\quad \text{SINR}_k(\qa, \ETAI, \ETAE) \geq \mathcal{T}_k, \forall k \in \K,~\label{ct:SINRThreshol2}
        \\
        &\quad \sum\nolimits_{k\in\K}{\etamkI} \leq a_{m}, \forall m \in \M~\label{ct:etamkI2}
        \\
        &\quad \sum\nolimits_{j\in\J}{\etamjE} \leq 1 - a_{m}, \forall m \in \M,~\label{ct:etamjE2}
        \\
        &\quad a_m \in \{0,1\},\forall m \in \M.~\label{ct:a_m2}
    \end{align}
\end{subequations}
We use \textbf{JAP-PA} to refer to the solution of joint AP mode selection and power control problem in~\eqref{eq:ProblemFormulationNL1}. Note that due to the binary nature of $a_m$, constraints~\eqref{ct:etamkI2} and~\eqref{ct:etamjE2} are equivalent to~\eqref{ct:etamkI}. Problem ($\mathcal{P}2$) is a complicated mixed-integer non-convex problem.   To address the integer binary constraint~\eqref{ct:a_m2}, it is straightforward to see that $a_m\in\{0,1\}$ is equivalent to $a_m=a_m^2$, $a_m\in[0,1]$. On the other hand, it holds true that $a_m^2\leq a_m$, for $a_m\in[0,1]$. Following~\cite{Che:TWC:2014}, we relax binary $a_m$ to $a_m\in[0,1]$ and introduce a penalty term with parameter $\lambda^{\mathrm{pen}}$ in the objective function to enforce $a_m=a_m^2$, thus making $a_m$ binary~\cite{Che:TWC:2014,Mohammadi:JSAC:2023}. We notice that this new constraint is not concave. Thus, we use the following concave lower bound
\vspace{-0.5em}
\begin{align}~\label{eq:ineq:x2}
    x^2 \geq x_0(2x - x_0).
\end{align}
To this end, we can recast~\eqref{eq:ProblemFormulationNL1} as
\vspace{-0.1em}
\begin{subequations}~\label{eq:ProblemFormulationP2}
    \begin{align}
        (\mathcal{P}2.1): & \quad \max_{\qa, \ETAI, \ETAE, \qe} \quad \sum\nolimits_{j\in\J}{e_{j}}  
        \nonumber\\ 
        &\hspace{0.5em} - \lambda^{\mathrm{pen}}\sum\nolimits_{m\in\M}{a_m - a_{m}^{(n)} \big(2a_m - a_{m}^{(n)}\big)},
        ~\label{obj:ObjectiveFunction2}
        \\
        \mathrm{s.t.} 
        &\quad 
        Q_j \geq \Xi(\tilde{e}_{j}),~\forall j \in \J,~\label{ct:auxilaryvariable2a}   
        \\
        &\quad e_j \geq \Gamma_j, \forall j \in \J,~\label{ct:auxilaryvariable2b2}
        \\
        &\quad \text{SINR}_k(\qa, \ETAI, \ETAE) \geq \mathcal{T}_k, \forall k \in \K,~\label{ct:SINRThreshol31}
        \\
        &\quad \sum\nolimits_{k\in\K}\!\!{\etamkI}\! \leq a_{m}^{(n)} \!\big(2a_m \!-\! a_{m}^{(n)}\big), \forall m \in \M~\label{ct:etamkI22}
        \\
        &\quad \sum\nolimits_{j\in\J}{\etamjE} \leq 1 - a_{m}^{2}, \forall m \in \M,~\label{ct:etamjE22}
        \\
        &\quad 0 \leq a_m \leq 1,\forall m \in \M,~\label{ct:a_m2relax2}
    \end{align}
\end{subequations}
where $\qe =\{e_j\geq 0, j\in\J\}$  are auxiliary variables, $\tilde{e}_{j} \triangleq (1 - \Omega) e_{j} + \phi \Omega$,  while superscript ($n$) denotes the value of the involving variable produced after $(n - 1)$ iterations ($n \geq 0$). Moreover,  $a_m$ has been replaced with $a^{2}_{m}$ in~\eqref{ct:etamkI2} and~\eqref{ct:etamjE2} to accelerate the convergence speed of the optimization problem. Note that the non-convex nature of~\eqref{ct:auxilaryvariable2a} and~\eqref{ct:SINRThreshol31} lead to the non-convexity of the problem~\eqref{eq:ProblemFormulationNL1}. To address the non-convex constraint~\eqref{ct:auxilaryvariable2a}, we use the convex upper bound of $\Xi(\tilde{e}_j)$ as
\vspace{-1em}
\begin{align}~\label{eq:CvxUBINV}
        \Xi(\tilde{e}_j) &\leq   \chi - \frac{1}{\xi} \bigg( \ln\Big(\frac{\phi - \tilde{e}_j}{\tilde{e}_j^{(n)}} \Big)
        - \frac{\tilde{e}_j - \tilde{e}_j^{(n)}}{\tilde{e}_j^{(n)}} \bigg)\triangleq \tilde{\Xi}(\tilde{e}_j).
\end{align}
Then, for the sake of notation simplicity, we define the expression holders $\mu_{mj} = (\Psi_1 + \Psi_2)/\alpha_{mj}^{\PMRT}$, $\varepsilon_{mj'} = (\Phi_1 + \Phi_2)/\alpha_{mj'}^{\PMRT}$, $\varrho_{mj} = \betamjeu \!+\! N \bar{\zeta}_{mj} \!+\! \frac{ \kappa}{L}\bar{\zeta}_{mj} \big\Vert \gmjeulos \big\Vert^{2}$, and
\begin{align}
    \varpi_{mj'} 
    &=  
    \frac{(L\!-\!K)^2}{\alpha_{mj'}^{\PMRT} L^2 } 
    \trace
    \bigg( \!\!
    \Big( \!
    \gameumjp \qI_{L} \!+\! \bar{\zeta}_{mj'} \kappa \gmjpeulos \! \big(\gmjpeulos \!\big)^\dag  
    \Big) \!
    \nonumber\\
    &\hspace{2em} \times
    \Big(
    \big(\betamjeu \!+\! N \bar{\zeta}_{mj} \big) \qI_{L}
    \!+\!
    \bar{\zeta}_{mj} \kappa \gmjeulos \big(\gmjeulos \!\big)^\dag
    \Big) \!\!
    \bigg), 
    \nonumber\\
    \mathcal{Z}_{mj} &= \sum\nolimits_{k\in\mathcal{K}} \etamkI \varrho_{mj} 
    - \etamjE \mu_{mj}
    \nonumber\\
    &- \sum\nolimits_{j'\in \Per \setminus j} \etamjpE \varepsilon_{mj'}
    - \sum\nolimits_{j' \notin \Per \setminus j}
     \etamjpE \varpi_{mj'}.
\end{align}
Now, we can recast~\eqref{ct:auxilaryvariable2a} as
\begin{align}~\label{eq:auxilaryvariable2a_convex}
    &4(\tau_c - \tau)\Snn\rho_{d} \sum\nolimits_{m\in\M}
    \Big(
    \etamjE \mu_{mj} 
    + \sum\nolimits_{j'\in \Per \setminus j} \etamjpE \varepsilon_{mj'}
    \nonumber\\
    &+ \sum\nolimits_{j' \notin \Per \setminus j}
     \etamjpE \varpi_{mj'}
    +\sum\nolimits_{k\in\mathcal{K}} \etamkI \varrho_{mj}
    \Big)
    \nonumber\\
    &+ (\tau_c - \tau)\Snn\rho_{d} \sum\nolimits_{m\in\M} (a_m + \mathcal{Z}_{mj})^{2} + 4(\tau_c - \tau)\Snn
    \nonumber\\
    &\geq 
    (\tau_c - \tau)\Snn\rho_{d} \sum\nolimits_{m\in\M} (a_m - \mathcal{Z}_{mj})^{2} 
    + 4\tilde{\Xi}(\tilde{e}_j),
\end{align}
which is still non-convex. We can replace the left-hand side of~\eqref{eq:auxilaryvariable2a_convex} by its concave lower bound, as
\begin{align}~\label{ct:auxilaryvariable2a_convex}
    &4(\tau_c - \tau)\Snn\rho_{d} \sum\nolimits_{m\in\M}
    \big(
    \etamjE \mu_{mj} 
    + \sum\nolimits_{j'\in \Per \setminus j} \etamjpE \varepsilon_{mj'}
    \nonumber\\
    &+ \sum\nolimits_{j' \notin \Per \setminus j}
     \etamjpE \varpi_{mj'}
    +\sum\nolimits_{k\in\mathcal{K}} \etamkI \varrho_{mj}
    \big)
    \nonumber\\
    &+ (\tau_c - \tau)\Snn\rho_{d} \sum\nolimits_{m\in\M} \big(a_m^{(n)} + \mathcal{Z}_{mj}^{(n)} \big) \times
    \nonumber\\
    &\hspace{1.5em}\Big(2\big(a_m + \mathcal{Z}_{mj} \big) - a_m^{(n)} - \mathcal{Z}_{mj}^{(n)} \Big)
    + 4(\tau_c - \tau)\Snn
    \nonumber\\
    &\geq 
    (\tau_c - \tau)\Snn\rho_{d} \sum\nolimits_{m\in\M} (a_m - \mathcal{Z}_{mj})^{2} 
    + 4\tilde{\Xi}(\tilde{e}_j),
\end{align}
where we used~\eqref{eq:ineq:x2} and replaced $x$ and $x_0$ by $a_m - \mathcal{Z}_{mj}$ and $a_m^{(n)} + \mathcal{Z}_{mj}^{(n)}$, respectively. Now, constraint~\eqref{ct:auxilaryvariable2a_convex} is convex. 

We now focus on~\eqref{ct:SINRThreshol31}, which is equivalent to
\begin{align}~\label{eq:ctSINR1}
    &\frac
    {\rho_{d}(L-K)}
    {\mathcal{T}_k}
    \Big(\sum\nolimits_{m\in\M}\sqrt{ a_m\etamkI  \gamuemk}  \Big)^2
    \nonumber\\
    &\geq
    \rho_{d} \sum\nolimits_{k'\in \Pir \setminus k} \! \Big( \! \sum\nolimits_{m \in \M}  \sqrt{\! (L-K) a_m \etamkpI \gamuemk} \Big)^{\!2}
    \nonumber\\
    &+
    \rho_{d} \sum\nolimits_{k'\in \K}  \sum\nolimits_{m \in \M}  a_m  \etamkpI \big( \betamkue \!-\! \gamuemk \big)
    \nonumber\\
    &+
    \rho_{d} \sum\nolimits_{j\in\mathcal{J}} \!\! \sum\nolimits_{m\in\MM} {(1-\! a_m\!)\etamjE(\betamkue \!-\! \gamuemk)}
    +  1.
\end{align}
We denote $\nu_{mk} \triangleq \betamkue - \gamuemk$, $\etamI \triangleq \sum\nolimits_{k'\in \K} \etamkpI$, and $\etamE \triangleq \sum\nolimits_{j\in\mathcal{J}}\etamjE$. By using~\eqref{eq:ineq:x2}, we replace the left-hand side of~\eqref{eq:ctSINR1} with its concave lower bound, which yields
\vspace{-0.2em}
\begin{align}~\label{eq:ctSINR2}
    &\frac
    {\rho_{d}(L-K) q_k^{(n)}}{\mathcal{T}_k}{\Big(2\sum\nolimits_{m\in\M}\sqrt{ a_m\etamkI  \gamuemk}  -q_k^{(n)}\Big)}    
    \nonumber\\
    &\geq
    \rho_{d} \sum\nolimits_{k'\in \Pir \setminus k} \! \Big( \! \sum\nolimits_{m \in \M}  \sqrt{\! (L-K) a_m \etamkpI \gamuemk} \Big)^{\!2}
    \nonumber\\
    &+
    \rho_{d}\sum\nolimits_{m \in \M} \nu_{mk}a_m \omega_m+
    \rho_{d}\sum\nolimits_{m\in\MM}\nu_{mk}  \etamE
    +  1,
\end{align}
where $q_k^{(n)}\triangleq \sum\nolimits_{m\in\M}\sqrt{ a_m^{(n)}\etamkIn  \gamuemk}$ and $\omega_m \triangleq \etamI-\etamE$. Constraint~\eqref{eq:ctSINR2} is still non-convex due to the $a_m \omega_m$ term. To deal with this, we apply $4a_m \omega_m = (a_m+\omega_m)^{2} - (a_m-\omega_m)^{2}$, and then utilize~\eqref{eq:ineq:x2} to get the following convex constraint:
\vspace{-0.2em}
\begin{align}~\label{eq:ctSINR3}
    &\frac
    {4\rho_{d}(L-K) q_k^{(n)}}{\mathcal{T}_k}{\Big(2\sum\nolimits_{m\in\M}\sqrt{ a_m\etamkI  \gamuemk}  -q_k^{(n)}\Big)}    
    \nonumber\\
    &+ \!
    \rho_{d}\!\sum\nolimits_{m\in\MM} \!\!\nu_{mk} (a_{m}^{(n)} -\!\omega_{m}^{(n)})\Big(2\big(a_{m} \!- \!\omega_{m} \big) \!- \!\big(a_{m}^{(n)} \!-\! \omega_{m}^{(n)}\big) \Big)
    \nonumber\\
    &\geq
    4\rho_{d} \sum\nolimits_{k'\in \Pir \setminus k} \! \Big( \! \sum\nolimits_{m \in \M}  \sqrt{\! (L-K) a_m \etamkpI \gamuemk} \Big)^{\!2}
    \nonumber\\
    &+ \rho_{d}\sum\nolimits_{m\in\MM} \nu_{mk} (a_m + \omega_m)^{2}
    +
    4\rho_{d}\sum\nolimits_{m\in\MM} \nu_{mk}\etamE + 4.
\end{align}
The optimization problem~\eqref{eq:ProblemFormulationNL1}  is now recast as
\begin{subequations}\label{opt:JAP:final}
\begin{alignat}{2}
&(\mathcal{P}2.2):  \max_{\qa, \ETAI, \ETAE, \qe}        
&\qquad&\sum\nolimits_{j\in\J}{e_{j}} \nonumber\\
&&&\hspace{-4em}
             - \lambda^{\mathrm{pen}}\!\sum\nolimits_{m\in\M}\!\!{a_m\! -\! a_{m}^{(n)} \big(2a_m \!-\! a_{m}^{(n)}\big)},
            ~\label{opt:JAP:final:obj}\\
&\hspace{5em}\text{s.t.} 
&         &~\eqref{ct:auxilaryvariable2a_convex},~\forall j \in \J,~\label{opt:JAP:final:ct1}\\
&         &      &~\eqref{eq:ctSINR3},~\forall k \in \K,
           ~\label{opt:JAP:final:ct2}\\
&         &      &~\eqref{ct:auxilaryvariable2b2},~\eqref{ct:etamkI22}-\eqref{ct:a_m2relax2}.~\label{opt:JAP:final:ct3}
\end{alignat}
\end{subequations}
Problem~\eqref{opt:JAP:final} is convex, and thus it can be solved using CVX \cite{cite:Grant:CVX}. In \textbf{Algorithm~\ref{alg1}}, we outline the main steps to solve problem ($\mathcal{P}2.2$), where $\widetilde{\qx} \triangleq \{\qa, \ETAI, \ETAE, \qe\}$ and $\widehat{\mathcal{F}} \triangleq\{\eqref{opt:JAP:final:ct1},~\eqref{opt:JAP:final:ct2},~\eqref{opt:JAP:final:ct3}\}$  is a convex feasible set. Starting from a random point $\widetilde{\qx}\in\widehat{\mathcal{F}}$, we solve \eqref{opt:JAP:final} to obtain its optimal solution $\widetilde{\qx}^*$, and use $\widetilde{\qx}^*$ as an initial point in the next iteration. The proof of this convergence property uses similar steps as  the proof of \cite[Proposition 2]{vu18TCOM}, and hence, is omitted herein due to lack of space.

\textbf{Complexity analysis:} In each iteration of \textbf{Algorithm~\ref{alg1}}, the computational complexity of solving  \eqref{opt:JAP:final} is $\OO(\sqrt{A_l+A_q}(A_v+A_l+A_q)A_v^2)$, since~\eqref{opt:JAP:final} can be equivalently reformulated as
an optimization problem with  $A_v\triangleq M(K+1)+J(M+1)$ real-valued scalar variables, $A_l\triangleq 2M+J$ linear constraints, and $A_q\triangleq M + K+ J+ 1$ quadratic constraints~\cite{tam16TWC}.

\begin{algorithm}[t]
\caption{Proposed algorithm for joint AP mode selection and power allocation design (\textbf{JAP-PA})}
\begin{algorithmic}[1]
\label{alg1}
\STATE \textbf{Initialize}: $n\!=\!0$, 
$\lambda^{\mathrm{pen}} > 1$, random initial point $\widetilde{\qx}^{(0)}\!\in\!\widehat{\mathcal{F}}$.
\REPEAT
\STATE Update $n=n+1$
\STATE Solve \eqref{opt:JAP:final} to obtain its optimal solution $\widetilde{\qx}^*$
\STATE Update $\widetilde{\qx}^{(n)}=\widetilde{\qx}^*$
\UNTIL{convergence}
\end{algorithmic}
\end{algorithm}
\setlength{\textfloatsep}{0.2cm}

\subsection{Learning-Based Solution  for ($\mathcal{P}2$) (\textbf{ML-JAP-PA})}~\label{sec:benchmarkII}
In this subsection, we propose an autonomous ML approach that alleviates the complexity of the SCA-based method introduced in Section~\ref{sec:jap-pa}. This approach adapts to and directly solves the non-convex problem $\mathcal{P}2$ over environmental interactions by learning from the evolving scenario.
\subsubsection{MDP-Based Problem Transformation}
We first transform the proposed system in Section~\ref{sec:Sysmodel} into a task for a RL agent - Markov Decision Process (MDP) framework. Markovian properties are defined as a tuple $\{\qs[t], \qa[t], r[t] \}$ where $\qs[t]$, $\qa[t]$, $r[t]$ are the observation, action, and reward at each $t \in {T}$, which is the number of total learning steps per episode.

At each time step $t$, the state $\qs[t]$, which encapsulates the system properties that fluctuate over time, is determined by the $x$-and-$y$ coordinates of the APs, IRs, ERs, and the BD-RIS. Then, the agent determines a joint action $\qa[t]$ contains the AP mode selection $\qa$, and power allocation $\ETAI$ and $\ETAE$. Systematically, the reward function $r[t]$ plays a crucial role for the RL agent in efficiently learning the task. The agent, based on the perceived state $\qs[t]$, executes an action $\qa[t]$, which results in an immediate reward $r[t]$ and observes the subsequent state $\qs[t+1]$. Then, the reward is part of a long-term accumulated value function $Q\big(\qs,\bar{\mu}(\qs)\big)$, which is used to `update' the policy $\bar{\mu} \colon  \mathcal{S} \rightarrow \mathcal{A}$. Policy  $\bar{\mu}$ is a map function guiding the agent to determine actions based on the current state to maximize future rewards. Mathematically speaking:
\begin{equation}
    \bar{\mu}^{*} = \argmax_{\bar{\mu}}  Q\big(\qs,\bar{\mu}(s)\big),
\end{equation}
where
\vspace{-1em}
\begin{equation}~\label{eq:Bellman}
    Q\big(\qs,\bar{\mu}(\qs)\big) = \Ex_{r,\qs}\bigg[ \sum\nolimits_{t=1}^{T}{\bar{\gamma}^{t-1}r[t] } \bigg],
\end{equation}
where $\bar{\gamma}$ is the discount factor. To this end, a meticulous design of the step reward function $r[t]$ will enhance the learning efficiency of the policy, thereby steering to the optimal system performance while satisfying the problem constraints.

\subsubsection{Constraint-Satisfied Reward Function Formulation}
We recast the objective function and its constraints in \eqref{eq:ProblemFormulationNL1} as the reward $r[t]$. We explore the $\textit{sigmoid}$ activation layer after the Ornstein–-Uhlenbeck (OU) exploration process \cite{Thien:DDPG:2023} to obtain the action value in range of $[0,1]$, thus ensuring to satisfy \eqref{ct:a_m2}. To satisfy the constraint \eqref{ct:etamkI2} and \eqref{ct:etamjE2}, we implement the $\textit{softmax}$ function for each power allocation of AP $m$, which is mathematically formulated as
\vspace{0.2em}
\begin{align}
    \etamkI &= \exp(\bar{\etamkI}) \Big(\exp\big( \sum\nolimits_{k' \in \mathcal{K}}{\bar{\etamkpI}} \big)^{-1}\Big), \text{ and},
    \nonumber\\
    \etamjE &= \exp(\bar{\etamjE}) \Big(\exp\big( \sum\nolimits_{j' \in \mathcal{J}}{\bar{\etamjpE}} \big)^{-1}\Big), \forall m \in \M.
\end{align}

To satisfy the constraints \eqref{ct:SINRThreshol2}, we first introduce the additive positive-valued penalty weights $\lambda_{\mathrm{SE}}$. Accordingly, we formulate the step reward function as
\begin{align}~\label{eq:MLreward}
    r[t] &= \sum\nolimits_{j\in\J}{\Ex\big\{\Phi^{\mathrm{NL}}_{j}\big( \qa[t] \big)\big\}} 
    - \lambda_{\mathrm{SE}}
    ,
\end{align}
where
\vspace{-0.5em}
\begin{align}
  \lambda_{\mathrm{SE}} 
  &\triangleq 
    \begin{cases}
      \lambda_{\mathrm{SE}} & \text{if } \SINRk\big(\qa[t]\big) < \mathcal{T}_{k}
      \\
      0 & \text{otherwise.} 
    \end{cases}    
\end{align}

It can be seen that the value $r[t]$ decreases significantly as \eqref{ct:SINRThreshol} is not satisfied and the Frobenius norm becomes significantly large. At each learning iteration, the reward value accumulates to \eqref{eq:Bellman}, which indicates the learning process of the RL agent to determine $\qa[t]$ that satisfies all the constraints in the problem~\eqref{eq:ProblemFormulationNL1}. To this end, the off-policy online RL method deep deterministic policy gradient (DDPG) is implemented for joint optimization.

We provide the pseudo-scheme for DRL-based optimization in \textbf{Algorithm~\ref{alg:ML}}. The preliminaries of the DDPG training process can be found in \cite[Section III.B]{Thien:DDPG:2023}. 

\vspace{0.5em}
\textbf{Complexity analysis:} The complexity of \textbf{ML-JAP-PA} comprises the workloads of the DNN models. We consider that a critic (actor) network has $L^{Q}$ ($L^{\bar{\mu}}$) hidden layers, each having an identical number of $n^{Q}$ ($n^{\bar{\mu}}$) neurons; $\vert \mathcal{S} \vert \triangleq 2(M+1) + 2(K +J)$ $\big(\vert \mathcal{A} \vert \triangleq M(1+K+J)\big)$ is the input (output) layer. 
The forward and backward passes in the backpropagation process at the critic network take $\mathcal{O}\big( E T \big(\vert \mathcal{S} \vert n^{Q} + L^{Q} (n^{Q})^{2} + n^{Q} \big)^2 \big)$. 
The complexity of training $E$ episodes over the mini-batch sample $B$ is $\mathcal{O}\big(E  T B \big( \big(\vert \mathcal{S} \vert n^{\bar{\mu}} + L^{\bar{\mu}} (n^{\bar{\mu}})^{2} + n^{\bar{\mu}} \vert \mathcal{A} \vert )^{2}
+ \big(
\vert \mathcal{S} \vert n^{Q} + L^{Q} (n^{Q})^{2} + n^{Q}
\big)^2
\big)\big)$. The sampling mini-batch, the "soft" updating of the target networks, and the OU exploration take $\mathcal{O}(E  T  B)$, $\mathcal{O}(ET (L^{\bar{\mu}} n^{\bar{\mu}} + L^{Q} n^{Q}) )$, and $\mathcal{O}(E  T  \vert \mathcal{A} \vert)$ computational resources, respectively. After the offline training, the complexity of online decision-making in the actor network is $\mathcal{O}(\vert \mathcal{S} \vert n^{\bar{\mu}} + L^{\bar{\mu}} (n^{\bar{\mu}})^{2} + n^{\bar{\mu}}  \vert \mathcal{A} \vert)$.

\begin{algorithm}[t]
\caption{Learning-based solution for (\textbf{ML-JAP-PA})}~\label{alg:ML}
    \begin{algorithmic}[1]
    \footnotesize
        \STATE \textbf{Initialize} the critic network $Q(s,a|\theta^{Q})$ and actor network $\bar{\mu}(s|\theta^{\mu})$ with stochastic weights $\theta^{Q}$ and $\theta^{\bar{\mu}}$, target network $Q'$ and $\bar{\mu}'$
        \FOR{episode $= 1,\ldots,E$}
            \STATE Initialize the Markovian  environment
            \FOR{step $= 1,\ldots,T$}
                \STATE \# Interacting:
                \STATE Observe $\qs[t]$, determine $\qa_{\bar{\mu}}[t]$ and receive $r[t]$ and $\qs[t+1]$
                \STATE Store $T$ experience tuples into buffer $B$
                \STATE \# Training:
                \STATE Compute target value $y[t] = r_{\sim B} + \gamma Q\left(s_{\sim B}',\bar{\mu} (s_{\sim B}'))|\theta^{Q}\right)$
                \STATE Update critic network by minimizing 
                \\
                $L(\theta^{Q}) = \mathbb{E}_{\sim B}\left[(Q(s_{\sim B},a_{\sim B}|\theta^{Q})-y(t))^2 \right]$
                \STATE Update actor network using the policy gradient method
                \STATE "Soft" update the target networks 
            \ENDFOR
        \ENDFOR
    \end{algorithmic}
\end{algorithm}
\setlength{\textfloatsep}{0.2cm}
\begin{Remark}~\label{remark:ML}
The trained policy network $\bar{\mu}^{*}$ instantaneously determines the optimal joint action of the AP-mode operation $\bar{\qa}^{*}_{m}$ and power allocation $\bar{\boldsymbol{\eta}}^{\mathtt{I}*}$ and $\bar{\boldsymbol{\eta}}^{\mathtt{E}*}$ at any given environmental state.
\end{Remark}

\vspace{-2em}
\subsection{Scattering Matrix Design}~\label{sec:phaseshift}
In this subsection, we introduce the designs for achieving a unitary scattering matrix for ($\mathcal{P}1$). 

\subsubsection{\HEU}
Now, we introduce a heuristic scheme to design a suboptimal $\boldsymbol{\Theta}$ given a network setup. As in Remark~\ref{remark:PC_NLHE_2}, we notice that $\boldsymbol{\Theta}$ directly affects the channel gain through the LoS components of the channels from the BD-RIS to $J$ ERs. Thus, we heuristically look for $\boldsymbol{\Theta}$ that obtains the highest magnitude of $\Vert \gmjeulos \Vert^{2}$ under the $D$ number of local network realizations. \textbf{Algorithm~3} provides the heuristic search for the suboptimal unitary and symmetric $\boldsymbol{\Theta}$ based on the projection in~\cite[III.B]{YMao_BDRIS_2023}. The complexity of this heuristic scheme is $\mathcal{O}(D (MJL + M N^{3} + M))$.

\subsubsection{\DFT}
We consider a widely used design for scattering matrix design, based of DFT,  to evaluate the performance of the proposed $\HEU$ scheme.  Hence, the  $(n,n')$-th element of the DFT scattering matrix is obtained as
\vspace{-1em}
\begin{align}
    [\boldsymbol{\Theta}]_{n,n'} = \frac{1}{\sqrt{N}} \bigg(\exp\Big(\frac{-j2\pi\bar{\theta}_{n}}{N}\Big) \bigg)^{i_H \times i_V},
\end{align}
where $i_H \triangleq \vert[N_{H}]_{n} - 1\vert$ and $i_V \triangleq \vert[N_{V}]_{n'} - 1\vert$ are the index values of the $n$-th and $n'$-th elements of the horizontal and vertical arrays, respectively. The DFT scheme is straightforward and immediately provides unitary and symmetric matrix.

\vspace{-0.4em}
\section{Benchmarks}~\label{sec:benchmark}
In this section, to assess the effectiveness of  \textbf{JAP-PA} scheme for CF-mMIMO systems, we introduce benchmark schemes for comparison in the numerical results of Section~\ref{sec:PerformanceAnalysis}.

\vspace{-1em}
\subsection{ Random AP Mode Selection and Power Control (\textbf{RAP-PA})} 
\vspace{-0.6em}
In this scheme, we assume that the AP modes ($\qa$) are randomly assigned. Accordingly, we optimize the power control coefficients ($\ETAI$, $\ETAE$), under the same SE requirement constraints for IRs and energy requirements for ERs. Therefore, the average sum-HE optimization problem is reduced to
\vspace{-0.1em}
\begin{subequations}~\label{eq:ProblemFormulationP3}
    \begin{align}
        (\mathcal{P}3): & \quad \max_{\ETAI, \ETAE, \boldsymbol{\ell}} \quad \sum\nolimits_{j\in\J}{\ell_{j}},~\label{obj:ObjectiveFunction32}
        \\
        \mathrm{s.t.} 
        &\quad 
        Q_j \geq \Xi(\tilde{\ell}_{j}),~\forall j \in \J,~\label{ct:auxilaryvariable32a}
        \\
        &\quad \ell_j \geq \Gamma_j, \forall j \in \J,~\label{ct:auxilaryvariable32b}
        \\
        &\quad \text{SINR}_k(\ETAI, \ETAE) \geq \mathcal{T}_k, \forall k \in \K,~\label{ct:SINRThreshol333}
        \\
        &\quad~\eqref{ct:etamkI2},~\eqref{ct:etamjE2},~\label{ct:etamjE3}
    \end{align}
\end{subequations}
where $\tilde{\ell}_{j} \triangleq (1-\Omega)\ell_j + \phi \Omega$ and $\ell_j$ is a new auxiliary variable. Problem~\eqref{eq:ProblemFormulationP3} is non-convex due to non-convex constraints  \eqref{ct:auxilaryvariable32a} and \eqref{ct:SINRThreshol333}.
In light of SCA, we replace the right-hand side of \eqref{ct:auxilaryvariable32a} with it convex upper-bound  to get the following convex constraint
\vspace{-0.1em}
\begin{align}~\label{eq:HE:SCA}
    &1\!+\! \rho_{d}\! \sum\nolimits_{m\in\M}{\etamjE \K_{1,mj}} 
    \!+\! \rho_{d} \!\sum\nolimits_{j'\in\Per}\!\!\sum\nolimits_{m\in\M} \!\etamjpE \K_{2,mj'}
    \nonumber\\
    &+\! \rho_{d} \sum\nolimits_{k\in\K}\sum\nolimits_{m\in\M} \!\etamkI \K_{3,mj}
    \!\geq\! 
    \tilde{\Xi}(\tilde{\ell}_j)/(\tau_c - \tau)\Snn,
\end{align}
where $\K_{1,mj} \triangleq \frac{1-a_m}{\alpha_{mj}^{\PMRT}}(\Psi_1 + \Psi_2)$, $\K_{2,mj'} \triangleq \frac{1-a_m}{\alpha_{mj'}^{\PMRT}}(\Phi_1 + \Phi_2)$, and $\K_{3,mj} \triangleq a_m \Big( \betamjeu \!+\! N \bar{\zeta}_{mj} \!+\! \bar{\zeta}_{mj} \frac{\kappa}{L} \big\Vert \gmjeulos \big\Vert^{2}  \Big)$.
\begin{figure*}[t]
    \centering
    \begin{subfigure}[b]{0.32\textwidth}
        \includegraphics[trim=2 0cm 0cm 0cm,clip,width=1.05\textwidth]{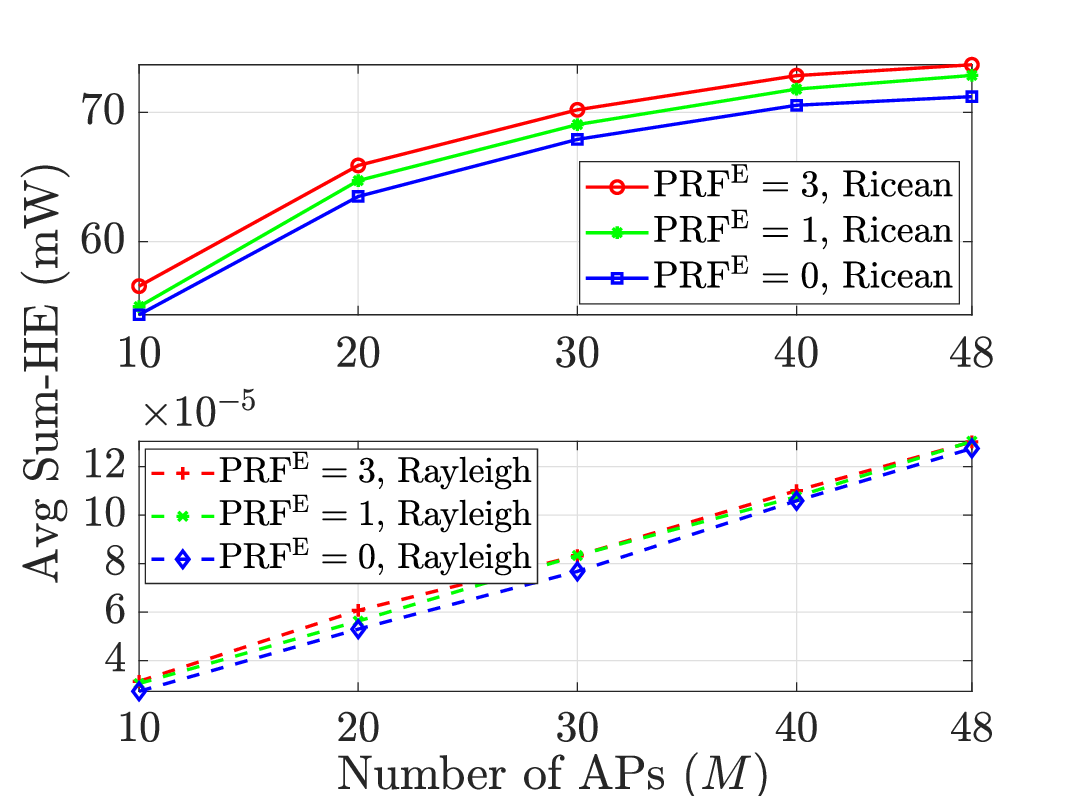}
        \caption{Channel model and pilot reuse factor}
        \label{fig:figurea}
    \end{subfigure}
    \begin{subfigure}[b]{0.32\textwidth}
        \includegraphics[trim=2 0cm 0cm 0cm,clip,width=1.05\textwidth]{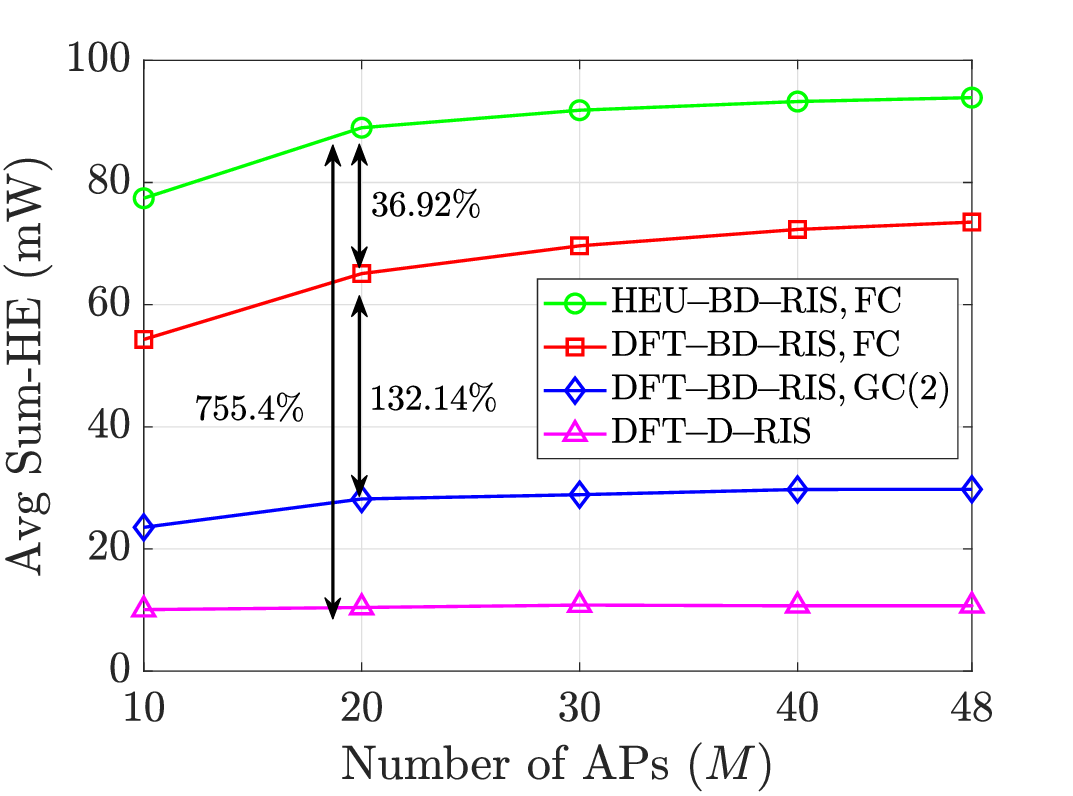}
        \caption{Impact of $M$, Ricean channel model }
        \label{fig:figureb}
    \end{subfigure}
    \begin{subfigure}[b]{0.32\textwidth}
        \includegraphics[trim=2 0cm 0cm 0cm,clip,width=1.05\textwidth]{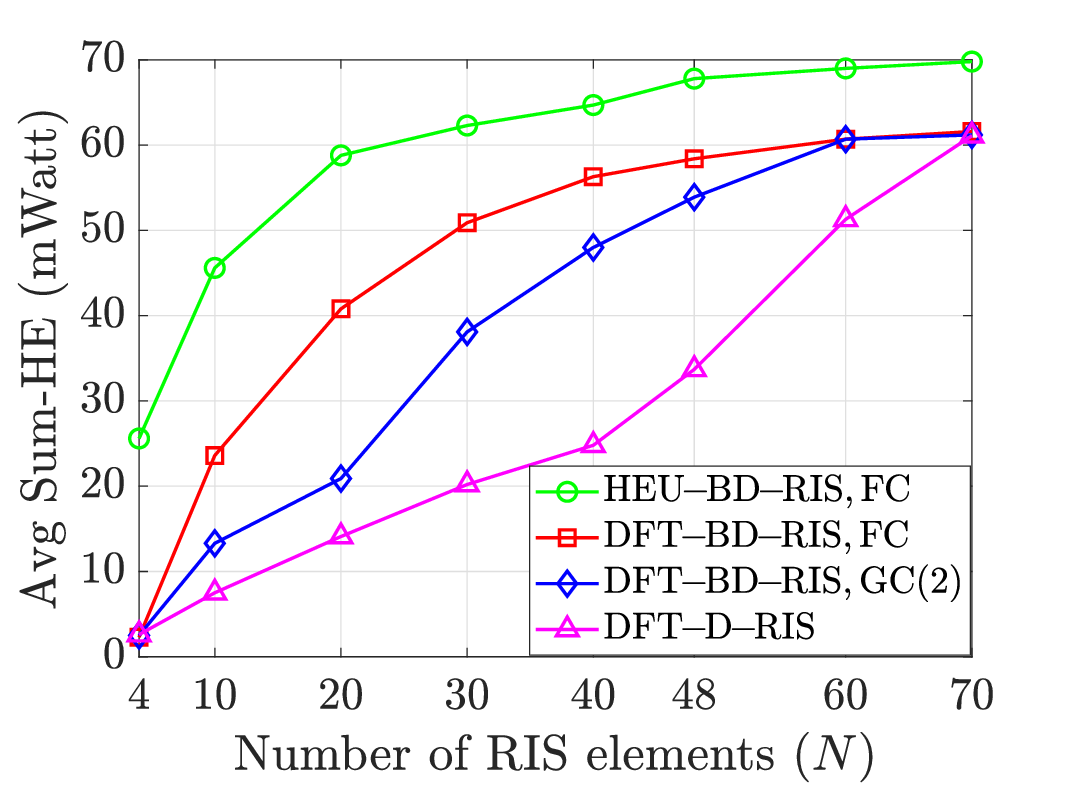}
        \caption{Impact of $N$, Ricean channel model}
        \label{fig:figurec}
    \end{subfigure}
    \caption{Impact of system parameters and  scattering matrix designs on the average sum-HE ($ML = 480, N=10, K = 3, J = 4$).}
    \vspace{-1em}
    \label{fig:three_figures}
\end{figure*}
\begin{algorithm}[t]
~\label{alg:heuristic_theta}
\caption{Heuristic Search for unitary and symmetric $\boldsymbol{\Theta}$}
    \begin{algorithmic}[1]
    \footnotesize
        \FOR{n = $1,...,D$}
            \FOR{m = $1,...,M$}
                \STATE Generate $\Hmlue \triangleq \{\hmlue \forall \J \}$, $\qZ_{m} \triangleq \{ \zmjlos \forall \J \}$
                \\ \hspace{1.9em} and $\FmRIS \triangleq \sqrt{\kappa\bar{\zeta}_{mj}} \FmRISlos + \sqrt{\bar{\zeta}_{mj}}\FmRISnlos$
                \STATE Compute $\qA \triangleq \FmRIS^\dag \Hmlue \qZ_{m}^\dag$
                \STATE Compute $\bar{\qA} \triangleq \frac{1}{2}(\qA + \qA^{T})$
                \STATE Obtain components $[\qU_{m}, \qV_{m}] = \mathrm{svd}(\bar{\qA})$ and rank $r$ of $\bar{\qA}$
                \STATE Partition $\qU_{m} = \big[\qU_{m}^{r_m}, \qU_{m}^{N-r_m}\big]$
                \\ \hspace{1.8em} and $\qV_{m} = \big[\qV_{m}^{r_m}, \qV_{m}^{N-r_m}\big]$
                \STATE Obtain $\THETA_{m}(n) \triangleq \hat{\qU}_m \qV^\dag$,  where $\hat{\qU}_m \triangleq [\qU_{m}^{r_m}, \big( \qV_{m}^{N-r_m} \big)^{*}]$
                \STATE Compute $\mathrm{obj}(n) = \sum\nolimits_{j \in \M} \sum\nolimits_{j \in \J} \Vert \gmjeulos \Vert^{2}$
                \STATE Store $\THETA_{m}(n), \mathrm{obj}(n)$ and its index $\forall m \in \M$
            \ENDFOR
        \ENDFOR
        \STATE Search the maximum $\mathrm{obj}(n^*)$ value and its index $n^*$
        \STATE Assign $n^*$-th $\boldsymbol{\Theta}$ as optimal $\boldsymbol{\Theta}^{*}, \forall m \in \M$
        \RETURN $\big[\boldsymbol{\Theta}^{*}\big]$
    \end{algorithmic}
\end{algorithm}
\setlength{\textfloatsep}{0.2cm}

Next, we focus on the constraint \eqref{ct:SINRThreshol3}, which is equivalent to
\vspace{-0.9em}
\begin{align}~\label{eq:RPC:SINR}
    {\rho_{d}(L-K)q_{k}^{(n)}}\Big( 2\sum\nolimits_{m\in\M}{\sqrt{\etamkI a_m \gamuemk}} - q_{k}^{(n)} \Big)
    \nonumber\\
    \hspace{4em}\geq \mathcal{T}_k(\PSI_{k}(\etamkI, \etamjE) + 1,
\end{align}
where we applied the concave upper bound in~\eqref{eq:ineq:x2}, in which $x^{(n)}$ is replaced with $q_{k}^{(n)} \triangleq \sum\nolimits_{m\in\M}{\sqrt{\etamkIn a_m \gamuemk}}$. Moreover, $\PSI_{k}(\etamkI, \etamjE)$ is a linear function of the variables $\etamkI$ and $\etamjE$, given by
\vspace{-0.3em}
\begin{align}
    &\PSI_{k}(\etamkI, \etamjE) = \rho_{d}
    \sum\nolimits_{k'\in \K} \sum\nolimits_{m \in \M}  \etamkpI a_m
    \big( \betamkue - \gamuemk \big)
    \nonumber\\
    &\hspace{1em}+ \rho_{d}
    \sum\nolimits_{j\in\mathcal{J}} \sum\nolimits_{m\in\M}{\etamjE(1 - a_m)(\betamkue \!-\! \gamuemk)}.
\end{align}

To this end, the convex approximation of \eqref{eq:ProblemFormulationP3} is given by
\vspace{-0.2em}
\begin{subequations}\label{opt:RPC:final}
\begin{alignat}{2}
&(\mathcal{P}3.1):  \max_{\ETAI, \ETAE, \boldsymbol{\ell}}        
&\qquad& \sum\nolimits_{j\in\J}{\ell_{j}}           ~\label{opt:RPC:final:obj}\\
&\hspace{2em}\text{s.t.} 
&     &~\eqref{eq:HE:SCA},~\forall j \in \J,~\label{opt:RPC:final:ct1}\\
&        &      &~\eqref{eq:RPC:SINR},~\forall k \in \K,
           ~\label{ct:SINRThreshol3}\\
&         &      &~\eqref{ct:auxilaryvariable32b},\eqref{ct:etamjE3},~\label{opt:RPC:final:ct3}
\end{alignat}
\end{subequations}
which can be efficiently solved using CVX.

\vspace{-0.3em}
\subsection{Random AP Mode Selection and Equal Power Allocation (\textbf{RAP-EPA})}
We introduce \textbf{RAP-EPA} as the benchmark with the lowest-complexity. We assume that the APs' operation mode selection parameters ($\aaa$) are randomly assigned and no power control is performed at the APs. Moreover, in the absence of power control, both E-APs and I-APs transmit at full power, i.e., at the $m$-th AP, power coefficients are the same and $\etamkI = \frac{1}{K}$, $\forall k\in\mathcal{K}$  and $\etamjE = \frac{1}{J}$, $\forall j\in\mathcal{J}$.

\vspace{-0.3em}
\subsection{Greedy AP Mode Selection and Equal Power Allocation (\textbf{GAP-EPA})}
We now propose a centralized greedy algorithm for AP selection that aims to maximize the sum-HE while satisfying the required QoS constraints. In this benchmark, all APs are initially configured to operate in information-transmission mode (I-AP), i.e., $\qa = \boldsymbol{1}_M$. Then, for each AP $m \in \mathcal{M}$, the algorithm sequentially attempts to switch its operating mode to energy-transmission mode (E-AP). After each tentative switch, the system evaluates whether the minimum HE and SE constraints in~\eqref{ct:SINRThreshol} and~\eqref{ct:HE_threshold} remain satisfied and whether the objective function improves. If both conditions are met, the new mode $a_m^*$ is accepted; otherwise, the AP's mode is reverted to I-AP mode.

\section{Numerical Examples}~\label{sec:PerformanceAnalysis}
In this section, we present numerical results to evaluate the performance of the proposed systems and to assess the effectiveness of the optimization design. We assume that the APs and IRs are randomly distributed in a square of $0.5 \times 0.5$ km${}^2$, whose edges are wrapped around to avoid the boundary effects. The height of the APs, BD-RIS, and IRs are 15~m, 15~m, and 1.65~m, respectively. We set $\tau_c = 200$, $\Snn = -92$ dBm, $\rho_d = 1$~W and $\rho_u = 0.1$ W.  
In addition, we set the non-linear EH parameters as $\xi = 150$, $\chi = 0.024$, and $\phi = 0.024$~W for the non-linear EH model. The LSFCs are generated following the three-slope propagation model from \cite{cite:HienNgo:cf01:2017} capturing the fading according to geometrical distance. Particularly, the fading factor follows a log-normal distribution with a standard deviation of $8$ dB. To illustrate various intrinsic levels of interference and channel errors, we denote the number of symbols for the UL channel estimation phase as $\tau = K + J - \mathrm{PRF}$, where $\mathrm{PRF} \triangleq \mathrm{PRF}^{\mathrm{I}} + \mathrm{PRF}^{\mathrm{E}}$ is the pilot reuse factor, which is clearly explained in \cite{Mohammadi:TC:2024}. Considering the vulnerability of information transmission to PC~\cite{mohammadi2024next}, we assume that $\tau>K$ and then assign the first $K$ out of $\tau$ orthogonal pilot sequences to the $K$ IRs as a priority (i.e., $\mathrm{PRF}^{\mathrm{I}}=0$), with the remaining $(\tau - K)$ pilot sequences then allocated to the ERs. We set the thresholds $\Gamma_{j} = 1 \text{ mW, }\forall j \in \J$ and $\SEth = 10$ bit/s/Hz,  $\forall k \in \K$, unless otherwise stated. We fix the mMIMO dimension to $ML=480$ for all simulations. In this way, various scenarios ranging from low to high density of APs can be analyzed and the performance comparison is fair.

\subsection{Impact of Network Parameters on the Average Sum-HE}~\label{section:numericalA}
We examine the impact of system parameters on the  performance of BD-RIS-aided CF-mMIMO SWIPT system in Fig.~\ref{fig:three_figures}. For unifying our analysis results, AP mode selection and power allocation are implemented under \textbf{RAP-EPA}. In our simulation results, we consider various RIS configurations, including fully-connected BD-RIS, two-group-connected BD-RIS, and diagonal RIS. These are denoted in this section as \textbf{BD-RIS, FC}, \textbf{BD-RIS, GC(2)}, and \textbf{D-RIS}, respectively. 

Figure~\ref{fig:figurea} shows the influence of the channel model and the PC on the performance of EH. For $J = 4$, we consider the three cases as $\mathrm{PRF}^{\mathrm{E}} = 0$, $\mathrm{PRF}^{\mathrm{E}} = 1$, and $\mathrm{PRF}^{\mathrm{E}} = 3$. As $\mathrm{PRF}$ increases,  the intrinsic interference among ERs also rises due to the degradation in CSI accuracy, which negatively impacts both the precoding and beamforming performance. To emphasize the performance in the HE zone, we keep $\mathrm{PRF}^{\mathrm{I}} = 0$ for a fair comparison. Figure~\ref{fig:figurea} shows that the HE reaches the milliwatt~(mW) range under the Ricean fading model, compared to only the microwatt~($\mu$W) range achieved with the Rayleigh model. Another important observation is that the scenario with $\mathrm{PRF}^{\mathrm{E}} = 3$ achieves the highest sum-HE. This highlights the benefit of sharing the same pilot sequences from an EH perspective, resulting in greater HE, which contrasts with the impact of pilot sequence sharing in information transmission applications. 

\begin{figure}[t]
	\centering
	\includegraphics[width=0.46\textwidth]{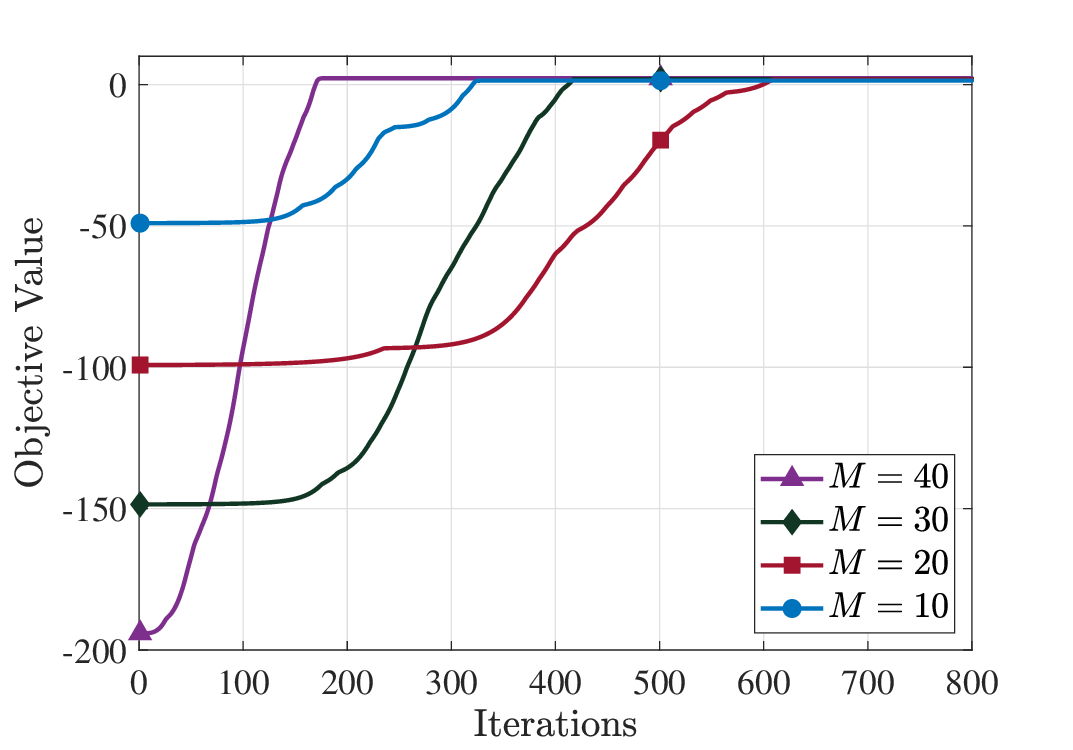}
	\caption{\small The convergence pattern of the \textbf{JAP-PA} solution ($ML = 480, N=10, \SEth = 10$ [bit/s/Hz], $ K = 3, J = 4$).}
	\label{fig:SCAPattern}
\end{figure}

\begin{figure}[t]
	\centering
	\includegraphics[width=0.46\textwidth]{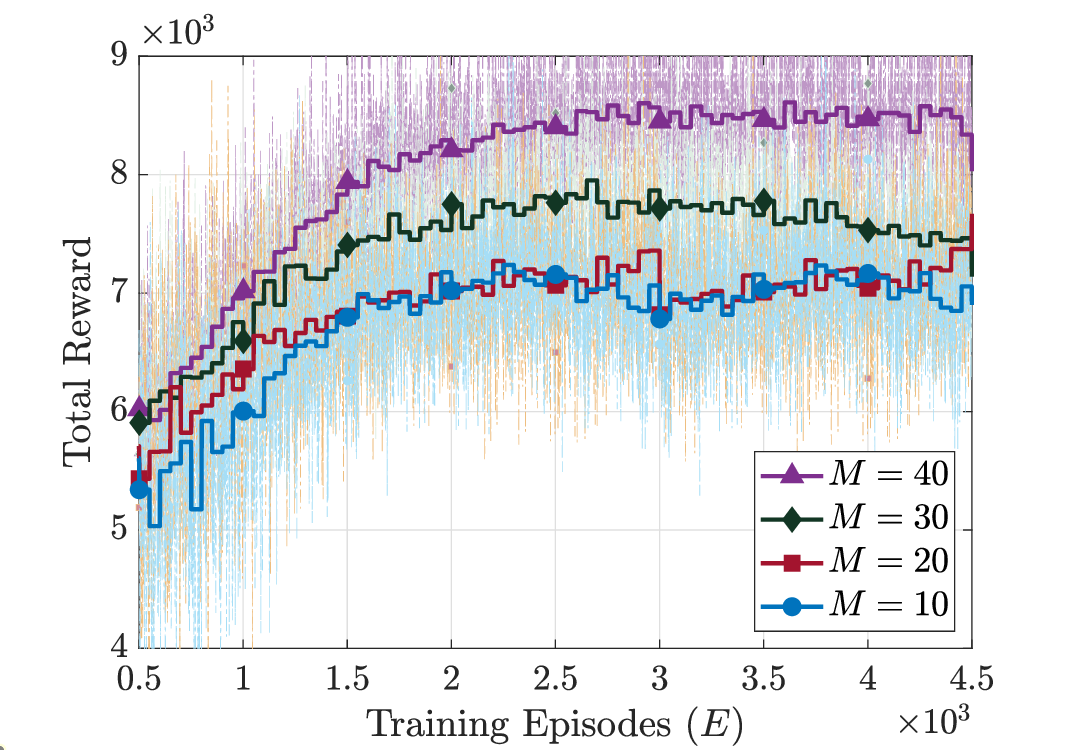}
	\caption{\small The convergence pattern of the \textbf{ML-JAP-PA} solution ($ML = 480, N=10, \SEth = 10$ [bit/s/Hz], $ K = 3, J = 4$).}
	\label{fig:trainingPattern}
\end{figure}

Figure~\ref{fig:figureb} shows the performance of different BD-RIS architectures in terms of the average sum-HE as a function of the number of APs. Moreover, the performance of the (\textbf{D-RIS}) has been included. The $\HEU$ design with the fully-connected (FC) architecture outperforms all other designs with a remarkable $755.4\%$ increase in the average sum-HE compared to the \textbf{D-RIS}, with this advantage growing as the number of APs increases. Additionally, the $\HEU$ design delivers up to a $36.92\%$ gain over the $\DFT$ architecture. Another key observation is that the FC enhances the performance by $132.14\%$ compared to GC architecture. These performance gains are attributed to the \textit{FC impedance network of BD-RIS}, which provides greater reconfigurability and more degrees of freedom in wave manipulation compared to the single-connected structure of D-RIS. These findings highlight the crucial role of phase shift element interconnection in boosting EH and underscore the importance of the scattering matrix as a key design parameter in EH applications.

Figure~\ref{fig:figurec} presents the average sum-HE as a function of the number of scattering elements. It is evident that $\textbf{HEU-BD-RIS, FC}$ consistently delivers the highest HE across all values of $N$. Additionally, $\textbf{DFT-BD-RIS, FC}$ sets an upper performance bound for both $\textbf{DFT-BD-RIS, GC(2)}$ and $\textbf{DFT-D-RIS}$, with the performance gap between these three DFT-based designs remaining negligible at $N = 4$ and $N = 70$. This analysis indicates the importance of carefully selecting the appropriate number of scattering elements and architecture to optimize the performance for specific scenarios.

\subsection{Proposed AP Selection and Power Control Designs}
In this subsection, we compare the performance of the proposed solutions. 
Figure~\ref{fig:SCAPattern} demonstrates the convergence behavior of the \textbf{JAP-PA} solution across different network setups. For $M = 40$, there is an abundance of serving APs to achieve the maximum sum-HE gains while the minimum SE threshold is satisfied. Thus, the scheme requires approximately 180 iterations to reach the optimal joint solution. Conversely, a reduced AP availability ($M = 10$) decreases the problem's dimensionality, which paradoxically lowers the computational complexity while introducing convergence trade-offs: although the algorithm achieves faster convergence, it settles at a suboptimal objective value. This illustrates the importance of network size and the complexity of the \textbf{JAP-PA}.

For \textbf{ML-JAP-PA}, we determine the hyper-parameters for the training process that fit the best with our studied systems through a trial-and-error process. The critic network is made up of two hidden layers (with 1028 and 512 nodes), while the actor network includes three hidden layers (with 1028, 512, and 256 nodes). For the training process of 3,000 episodes and 500 steps per episode, we set the batch size $B= 128$, discount factor $\bar{\gamma} = 0.99$, and buffer length $= 1e6$. The decay rates of $\epsilon$ and the update of the target network are 1e-3. The learning rates for the actor and critic network are 1e-4 and 5e-3, respectively. ReLu activation is implemented for all hidden layers, while Tanh and Sigmoid activations are applied for the actor's decision-making and action-clipping post OU noise, respectively. Figure~\ref{fig:trainingPattern} illustrates the convergence of the \textbf{ML-JAP-PA} solution. As the network scales, the agent earns progressively higher total rewards, resulting in improved sum-HE performance, which will be further validated in the subsequent figures.

\begin{table}[t]
\caption{\small Execution time of the optimization solutions (seconds per iteration)} 
\vspace{-1em}
\footnotesize
\begin{center}
\renewcommand{\arraystretch}{1}
\begin{tabular}{|r|r|r|r|r|}
\hline
\multicolumn{1}{|c|}{\multirow{2}{*}{\textbf{Benchmarks}}} & \multicolumn{4}{c|}{\textbf{Number of APs ($M$)}}                                                \\ \cline{2-5} 
\multicolumn{1}{|c|}{}                                   & \multicolumn{1}{c|}{\textbf{$10$}} & \multicolumn{1}{c|}{\textbf{$20$}} & \multicolumn{1}{c|}{\textbf{$30$}} & \multicolumn{1}{c|}{\textbf{$40$}} \\ \hline
\textbf{JAP-PA}      & 178.17   & 338.23    & 482.41   & 626.74       \\ \hline
\textbf{RAP-PA}       & 1.9501   & 3.6956    & 3.7770   & 5.6361      \\ \hline
\textbf{GAP-PA}      & 0.0720   &  0.1160    & 0.1801   & 0.3218       \\ \hline
\textbf{ML-JAP-PA} & 0.0378   & 0.0401   & 0.0420    &  0.0472       
\\ \hline
\textbf{RAP-EPA}       & 0.0042   & 0.0043    & 0.0045    & 0.0050    
\\ \hline

\end{tabular}%
\end{center}
\label{table:execution_time}
\end{table}

\begin{figure}[t]
	\centering
	\includegraphics[width=0.44\textwidth]{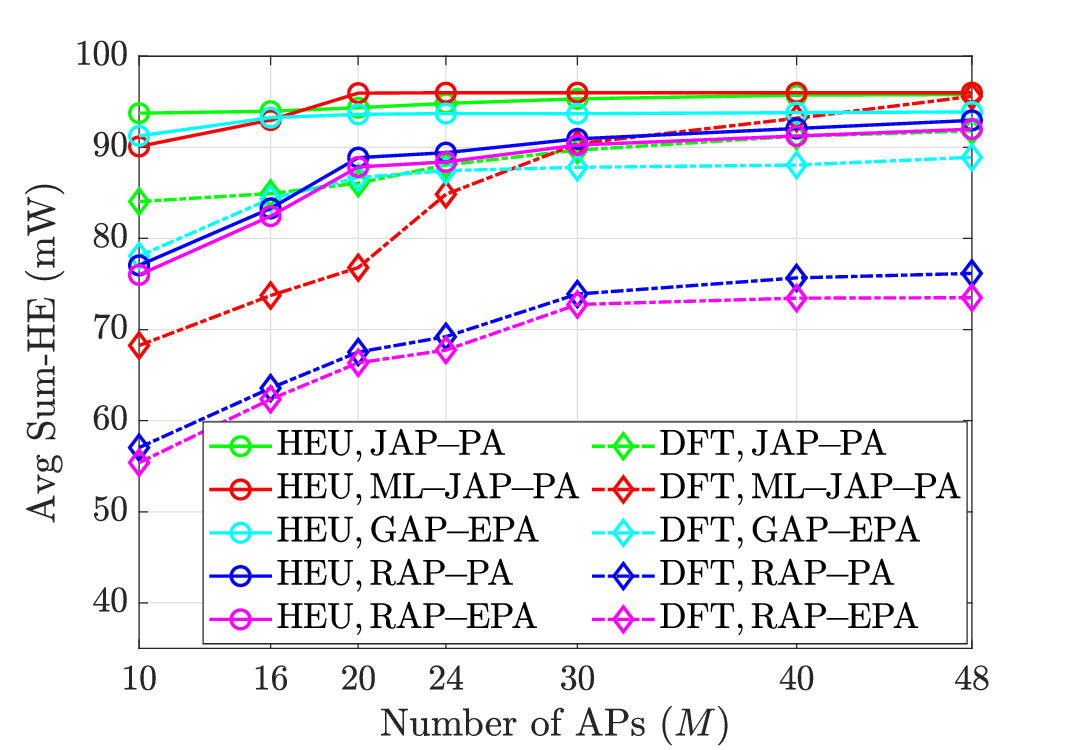}
	\caption{\small Performance of the proposed optimization solutions versus the number of APs ($ML = 480, N=10, \SEth = 10$ [bit/s/Hz], $K = 3, J = 4$, $\mathrm{PRF}^{\mathrm{I}} = 0$, $\mathrm{PRF}^{\mathrm{E}} = 3$).}
	\label{fig:HEvMbenchmarks}
\end{figure}

Table~\ref{table:execution_time} shows the execution time in [second/iteration] of the proposed solutions. We observe that the execution time of \textbf{JAP-PA} increases proportionally with the scale of the architecture, while the post-training execution time of \textbf{ML-JAP-PA} is nearly negligible regardless of the network scalability, and reduces the iteration time by $99.99\%$ at $M=40$. This highlights the advantage of real-time execution of the ML-based solution.

\begin{table}[t!]
 \small
\caption{Achievable minimum SE required per-IR } 
\vspace{-1em}
\footnotesize
\begin{center}
\renewcommand{\arraystretch}{1}
\begin{tabular}{|>{\centering\arraybackslash}m{2.28cm}|>{\centering\arraybackslash}m{0.77cm}|>{\centering\arraybackslash}m{0.77cm}|>{\centering\arraybackslash}m{0.77cm}|>{\centering\arraybackslash}m{0.77cm}|>{\centering\arraybackslash}m{0.77cm}|}
\hline
\multicolumn{1}{|c|}{\multirow{2}{*}{\textbf{Benchmarks}}} & \multicolumn{5}{c|}{\textbf{Minimum SE required per IR ($\SEth$ (bit/s/Hz))}}                                                \\ \cline{2-6} 
\multicolumn{1}{|c|}{}                                   & \multicolumn{1}{c|}{\textbf{$8$}} & \multicolumn{1}{c|}{\textbf{$10$}} & \multicolumn{1}{c|}{\textbf{$12$}} & \multicolumn{1}{c|}{\textbf{$15$}} & \multicolumn{1}{c|}{\textbf{$18$}}
\\ \hline
\textbf{HEU, JAP-PA}      & 10.28   & 11.34    & 12.66   & 15.51  & 18.54      \\ \hline
\textbf{DFT, JAP-PA}      & 14.17   & 14.70    & 14.80   & 16.47  & 17.02      \\ \hline
\textbf{HEU, GAP-EPA}      & 8.83   & 10.96    & 12.33   & 15.10  & 19.40      \\ \hline
\textbf{DFT, GAP-EPA}      & 8.93   & 10.87    & 14.62   & 15.04   & 17.06      \\ \hline
\textbf{HEU, ML-JAP-PA}      & 15.67   & 15.63    & 15.88   & 16.46  & 18.15      \\ \hline
\textbf{DFT, ML-JAP-PA}      & 14.16   & 14.21    & 14.34   & 16.54  & 17.34   \\ \hline
\textbf{HEU, RAP-PA}      & 17.41   & 17.35    & 17.44    & 17.33  & 16.22    \\ \hline
\textbf{DFT, RAP-PA}      & 17.02   & 17.14    & 17.61    & 17.38  & 14.73     \\ \hline
\textbf{HEU, RAP-EPA}      & 17.24   & 17.22    & 17.29    & 16.89  &  9.17     \\ \hline
\textbf{DFT, RAP-EPA}      & 16.80   & 16.94    & 17.37    & 16.26  &  7.42      \\ \hline
\end{tabular}%
\end{center}
\vspace{-0.5em}
\label{table:minSE}
\end{table}

Figure~\ref{fig:HEvMbenchmarks} compares the average sum-HE across different designs as the number of APs varies. Similar to the pattern observed in Fig.~\ref{fig:figureb} implementing the BD-RIS heuristic scattering optimization results in improvements of $12.4\%$, $20.8\%$, $13.9\%$, $22.1\%$, and $22.7\%$ over the low-complexity DFT BD-RIS for \textbf{JAP-PA}, \textbf{ML-JAP-PA}, \textbf{GAP-EPA}, \textbf{RAP-PA}, and \textbf{RAP-EPA}, respectively. Additionally, while the performance of \textbf{DFT, ML-JAP-PA} decreases proportionally as $M$ decreases, the heuristic approach enables \textbf{ML-JAP-PA} to achieve performance levels close to those of the proposed \textbf{JAP-PA} for $M > 20$. Another noteworthy observation is that \textbf{RAP-PA} outperforms \textbf{RAP-EPA} by only about $5\%$ for both heuristic and DFT scattering matrix designs, underscoring the significance of AP mode selection. In addition, we observe a saturation behavior in the sum-HE metric when the number of APs ($M$) becomes sufficiently large. It is important to recall that, due to the inherent non-linear characteristics of the NLEH model, the maximum energy that can be harvested per ER is 24 mW. As $M$ increases, the distance between APs and receivers decreases, allowing the SE QoS requirements to be satisfied with fewer APs dedicated to WIT. This frees up more APs for WPT, which, combined with the reduced path loss, boosts the harvested energy at the EUs and accelerates the rate at which they reach the saturation point. Therefore, we observe that with a denser AP deployment, each AP can operate at a lower transmit power, reducing the overall energy consumption of the network.

\begin{figure}[t]
	\centering
	\includegraphics[width=0.44\textwidth]{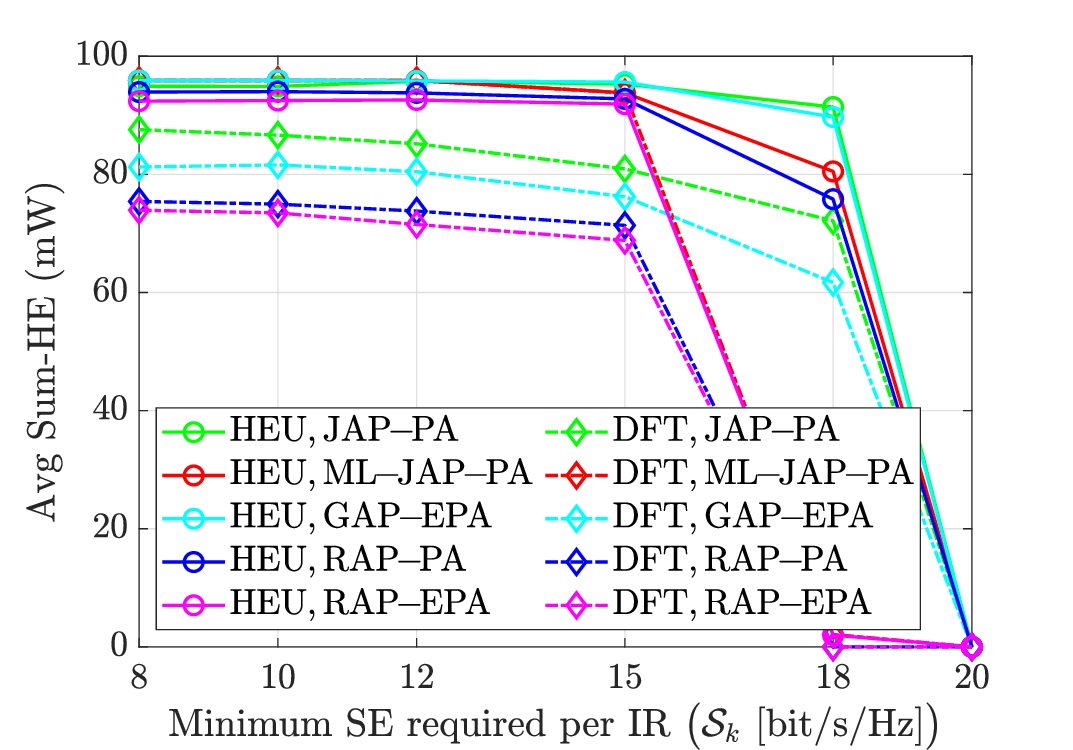}
	\caption{\small Performance of the proposed optimization solutions versus the minimum required SE per IR ($M = 40, L = 12, N=10, K = 3, J = 4$, $\mathrm{PRF}^{\mathrm{I}} = 0$, $\mathrm{PRF}^{\mathrm{E}} = 3$).}
	\label{fig:HEvSEbenchmarks}
\end{figure}

Figure~\ref{fig:HEvSEbenchmarks} illustrates the sum-HE gain verus the QoS SE requirements $\SEth$, and Table~\ref{table:minSE} lists the achievable min $\SEk$ values under increasing $\SEth$ values. We first observe a general trend: as the SE QoS requirement $\SEk$ increases, more APs are allocated to WIT. This, in turn, reduces the number of APs available for WPT, resulting in a significant drop in the sum-HE curves. Notably, while the benchmark schemes experience this decline at $\SEk > 15$ [bit/s/Hz], our proposed solutions can effectively support SE requirements exceeding $19$ [bit/s/Hz]. In particular, we observe that \textbf{JAP-PA} effectively allocates more I-APs to meet the minimum $\SEth$ requirement, demonstrating its adaptability to varying SE thresholds. In contrast, for \textbf{ML-JAP-PA}, the optimal solution at $\SEth \!=\! 18$~[bit/s/Hz] is achieved when paired with the heuristic scattering matrix but fails when combined with the DFT design. This can be attributed to the cardinal disadvantage of DRL-based solutions, which are \textit{penalty sensitive}. As $\SEth$ increases, $r[t]$ are penalized with $\lambda_{\mathrm{SE}}$, leading to instability reward patterns that greatly disturb the training process. Since \textbf{GAP-EPA} initially configures all APs as I-APs, it can readily satisfy the low $\mathcal{S}_k$ threshold and achieves a sum-HE gain only slightly lower than that of \textbf{JAP-PA}. Nevertheless, \textbf{GAP-EPA} achieves almost 10\% lower in sum-HE gain compared to \textbf{JAP-PA} in high QoS SE scenario due to the lack of transmit power allocation schemes. We can also observe that, due to the lack of optimal solution for $\qa$, both \textbf{RAP-PA} and \textbf{RAP-EPA} cannot allocate extra APs for WIT operation, making them unable to meet the requirements in high $\SEth$ scenarios. Overall, due to the heuristic optimization of the BD-RIS scattering matrices, we observe an approximate 25\% increase in sum-HE.

\begin{figure}[t]
	\centering
		\includegraphics[width=0.44\textwidth]{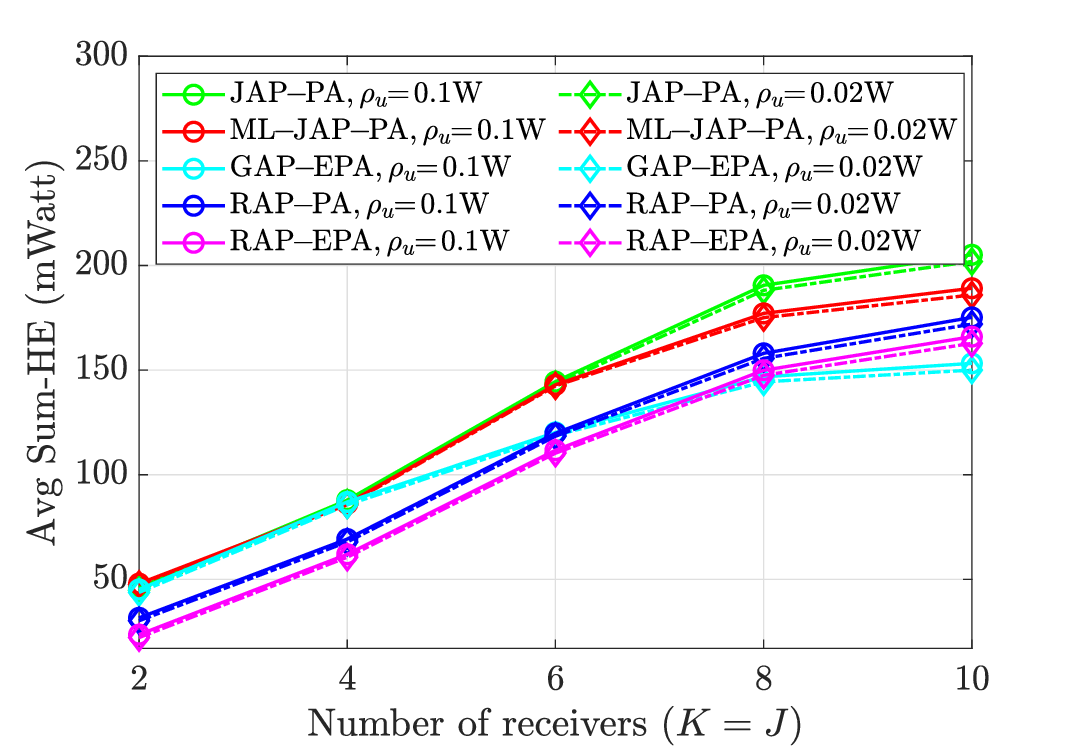}
		\caption{\small Performance of the proposed optimization solutions versus the number of receivers (\textbf{DFT-BD-RIS}, $ML = 480, M=30, \SEth = 10$ [bit/s/Hz], $\mathrm{PRF}^{\mathrm{I}} = 0$, $\mathrm{PRF}^{\mathrm{E}} = J-1$).}
		\label{fig:SHEvReceivers}
\end{figure}

Figure~\ref{fig:SHEvReceivers} illustrates the impact of pilot power $\rho_{u}$, PRF, and receiver scales on the sum-HE metric. We note that the channel estimation accuracy deteriorates as the pilot power decreases. By examining a range of pilot power values, we effectively demonstrate how the channel estimation errors influence the SWIPT performance. Furthermore, as the number of receivers increases, the NLEH gain per ER gradually decreases due to the fixed number of APs. Nevertheless, thanks to the effective joint optimization strategies, the overall sum-HE increases significantly, regardless of variations in the channel estimation accuracy. Although a lower $p_{u}$ increases the norm of the statistical channel estimation errors, the average NLEH value remains largely unaffected due to the consideration of joint optimization solutions. By scaling the number of receivers, we observe a significant impact on the achievable SE and HE gains when optimizing the power allocation schemes. Since \textbf{GAP-EPA} only searches for the best set of $\qa^{*}$ without optimizing $\ETAI$ and $\ETAE$, the WPT operation improves slightly in scenarios with a large number of receivers, resulting in the lowest sum-HE gain when serving 20 receivers. These numerical results clearly demonstrate that our proposed algorithms effectively maintain high performance and scalability, successfully addressing practical deployment challenges across a wide range of network sizes and conditions.

To conclude this subsection, we summarize the following key observations and its intrinsic laws based on the numerical analysis: \textit{i)} For both cases of heuristic and DFT scattering matrix, the proposed \textbf{JAP-PA} yields the best performance despite the long execution time. \textit{ii)} \textbf{ML-JAP-PA} can effectively assign I-APs and E-APs but becomes inefficient in implementing power control, leading to weaker performance compared to \textbf{JAP-PA}. \textit{iii)} The ML-based benchmark performs well in LoS-dominant scenarios, given the specified training hyper-parameters. However, without the heuristic scattering matrix to establish dominant LoS links, \textbf{ML-JAP-PA} shows suboptimal performance in scenarios with a lower number of APs and high $\SEth$, raising concerns about the consistency and adaptability of the ML-based approach. \textit{iv)} By jointly optimizing the AP mode selection $\qa$, the average sum-HE at the ERs achieves an approximately $20\%$ gain, thus proving its importance in our proposed CF-mMIMO SWIPT system.

\vspace{-0.5em}
\section{Conclusion}~\label{sec:conclusion}
We developed a comprehensive framework to optimize the interconnected scattering matrix, along with AP mode selection parameters, and AP power allocation coefficients for a BD-RIS-assisted CF-mMIMO SWIPT system. Utilizing closed-form expressions for the average sum-HE of ERs and the achievable SE of IRs, we formulated and solved a non-integer non-convex optimization problem via the robust SCA-based and autonomous DRL-based solutions. Our simulation results demonstrated that the proposed design significantly improves the SWIPT performance thanks to the proposed optimization schemes. Furthermore, our performance analysis revealed the advantages of PC in enhancing performance of EH applications. The integration of BD-RISs with the proposed optimization solutions is numerically shown to be effective in enhancing the SWIPT performance for long-term 6G networks. 

The practical relevance of our work is evident in scenarios such as smart factories, where CF-mMIMO provides seamless coverage and SWIPT enables wireless power delivery to large-scale IoT sensor networks. To address scalability challenges in ultra-dense deployments, a user-centric adaptation of the proposed design represents a promising direction for future research. Additionally, synchronization remains a critical challenge that must be addressed to fully avail of the potential of such systems.
Finally, future work could explore hybrid RIS architectures~\cite{hydrid:RIS}, distributed DRL frameworks deployed at the APs, and expand~\cite{cite:hua2024BdrisSwipt} to investigate the incorporation of stacked intelligent metasurfaces (SIM) and implement broader numerical/physical comparisons across BD-RIS, SIM and other RIS architectures.

\vspace{-0.5em}
\appendices

\section{Useful Derivations}\label{appendix:usefulDerivations}

We first present the following lemmas to facilitate further derivations. These terms are a function of individual components of the statistical CSI-based design, such as fading, interference, and path loss.

\begin{Lemma}~\label{lemma:1}
Leveraging the statistical independence of the direct and indirect links, the second moment of the aggregated RIS-assisted actual and estimated channels from AP $m$ to ER $j$ can be respectively obtained as:
\begin{align}
    \Ex\big\{ \big\Vert\gmjue \big\Vert^2 \big\} 
    &=  L \betamjeu \!+\! LN \bar{\zeta}_{mj} + \bar{\zeta}_{mj} \kappa (\gmjeulos)^\dag \gmjeulos.
\end{align}
Moreover, the second-order and fourth-order moments of the estimate of the channels towards ERs, i.e., $\Ex\big\{ \big\Vert\hgmjue \big\Vert^2 \big\}$ and $\Ex\big\{ \big\Vert\hgmjue \big\Vert^4 \big\}$ are given by
\begin{subequations}
 \begin{align}
    \Ex\big\{  \big\Vert\hgmjue \big\Vert^2  \big\} 
    &= L \gameumj + \bar{\zeta}_{mj} \kappa \big\Vert \gmjeulos \big\Vert^2,~\label{eq:2ndmoment_hgmjeu_jj}
\\ 
     \Ex\big\{  \big\Vert\hgmjue \big\Vert^4  \big\} &=( \bar{\zeta}_{mj} \kappa)^{2} \Vert \gmjeulos \Vert^4 \nonumber\\
    &\hspace{-4em}+ 2(L+1) \gameumj \bar{\zeta}_{mj} \kappa \big\Vert \gmjeulos \big\Vert^{\!2} +L(L +1)(\gameumj)^2.~\label{eq:4ndmoment_hgmjeu_jj}
\end{align}     
\end{subequations}
\end{Lemma}

\begin{proof}
The second moment of the NLoS term in \eqref{eq:aggregate_indirect_link} is computed as
\vspace{-0.5em}
\begin{align}
    &\Ex \Big\{ \!
    \Big( \sqrt{\bar{\zeta}_{mj}} \qz_{j}^\dag \boldsymbol{\Theta}^\dag \!(\FmRISlos)^\dag \!+\! \sqrt{\frac{\bar{\zeta}_{mj}}{\kappa}} \qz_{j}^\dag \boldsymbol{\Theta}^\dag \!(\FmRISnlos)^\dag \Big)  \nonumber\\
    &\hspace{2em} \times \Big( \sqrt{\bar{\zeta}_{mj}} \FmRISlos \boldsymbol{\Theta} \zmjlos + \sqrt{\frac{\bar{\zeta}_{mj}}{\kappa}} \FmRISnlos \boldsymbol{\Theta} \qz_{j} \Big) 
    \Big\} \nonumber
\end{align}
\begin{align}
    &\stackrel{(a)}{=}
    \bar{\zeta}_{mj} \qz_{m}^\dag \boldsymbol{\Theta}^\dag (\FmRISlos)^\dag \FmRISlos \boldsymbol{\Theta} \qz_{j} \nonumber\\
    &\hspace{2em} + 
    \frac{\bar{\zeta}_{mj}}{\kappa} 
    \trace \Big( \boldsymbol{\Theta}^\dag \Ex \big\{ \!(\FmRISnlos)^\dag \FmRISnlos \big\} \boldsymbol{\Theta} \qz_{j} \qz_{j}^\dag  \Big)\nonumber
    \\
    &= \bar{\zeta}_{mj} \big\Vert \gmjeulos \big\Vert^2 + LN \frac{\bar{\zeta}_{mj}}{\kappa}, 
\end{align}
where the second term of (a) exploits $\qx^\dag \qy = \trace(\qy \qx^\dag)$, the elements of $\FmRISnlos$ are i.i.d. $\CN(0,1)$, the unitary constraint $\boldsymbol{\Theta}^\dag \boldsymbol{\Theta} = \qI_L$, and  $[\qz_{j} \qz_{j}^\dag]_{n,n} = 1$. Then, combining this result with the second-order moment of the LoS term in \eqref{eq:aggregate_indirect_link}, the desired result in~\eqref{eq:2ndmoment_hgmjeu_jj} is obtained.

The fourth-order moment can be expressed as 
\begin{align}~\label{eq:4thmoment_hgmjeu_jj}
    &\Ex\big\{ \big\Vert\hgmjue \big\Vert^4 \big\} 
    = \Ex\Big\{ \Big\Vert \bar{\zeta}_{mj} \kappa \big (\gmjeulos)^\dag \gmjeulos    
    + 
    \sqrt{\!\bar{\zeta}_{mj} \kappa} (\gmjeulos )^\dag  \tildehgmjue \nonumber\\
    &\hspace{5.5em} + \sqrt{\bar{\zeta}_{mj} \kappa} (\tildehgmjue)^\dag \gmjeulos 
    + (\tildehgmjue)^\dag \tildehgmjue
    \Big\Vert^{2} \Big\}  
    \nonumber\\
    &\! = \! \Ex \Big\{ \!
    ( \bar{\zeta}_{mj} \kappa)^{2} \Vert \gmjeulos \Vert^4 
    \!+\!
    \Vert \tildehgmjue \Vert^4
    \!+\! 2\bar{\zeta}_{mj} \kappa (\gmjeulos\!)^\dag \gmjeulos  (\tildehgmjue)^\dag \tildehgmjue  
    \nonumber\\
    &+ \bar{\zeta}_{mj} \kappa  (\tildehgmjue \!)^\dag \gmjeulos (\gmjeulos \!)^\dag \tildehgmjue 
    + \bar{\zeta}_{mj} \kappa (\gmjeulos \!)^\dag \tildehgmjue (\tildehgmjue \!)^\dag \gmjeulos\Big\}.  
   \end{align}
Now, by exploiting the property $\Ex\{ \qx^\dag \qy \} = \Ex\{ \trace (\qy \qx^\dag)$ and noticing that $\tildehgmjue$ is the zero-mean NLoS component of $\hgmjue$, i.e., $\tildehgmjue \sim \CN(\boldsymbol{0}, \gameumj \qI_{L})$, we get
\begin{align}
    &\Ex\big\{ \big\Vert\hgmjue \big\Vert^4 \big\} = ( \bar{\zeta}_{mj} \kappa)^{2} \Vert \gmjeulos \Vert^{4} 
    \!+\! 
    L(L \!+\! 1)(\gameumj)^2 \nonumber\\
    &\hspace{3em} + 
    2\bar{\zeta}_{mj} \kappa \big\Vert \gmjeulos \big\Vert^{2} \Ex\big\{ (\tildehgmjue )^\dag  \tildehgmjue \big\} 
    \nonumber\\
    &\hspace{3em}
    + \bar{\zeta}_{mj} \kappa (\gmjeulos )^\dag  \Ex \big\{  \tildehgmjue (\tildehgmjue )^\dag \big\} \gmjeulos 
    \nonumber\\
    &\hspace{3em}
    +\bar{\zeta}_{mj} \kappa \trace\Big(\gmjeulos (\gmjeulos)^\dag \Ex \big\{ \tildehgmjue (\tildehgmjue)^\dag \big\}  \Big),
    \nonumber\\
    &= ( \bar{\zeta}_{mj} \kappa)^{2} \Vert \gmjeulos \Vert^4 
    + L(L + 1)(\gameumj)^2 + 2L \gameumj \bar{\zeta}_{mj} \kappa \big\Vert \gmjeulos \big\Vert^{2} \nonumber\\
    &+ \gameumj \bar{\zeta}_{mj} \kappa (\gmjeulos )^{H} \qI_{L} \gmjeulos + \gameumj \bar{\zeta}_{mj} \kappa \trace \Big(\gmjeulos (\gmjeulos )^{H} \qI_{L} \Big),
    \nonumber
\end{align}
where the final equality invokes the fourth-moment derivation of a zero-mean Gaussian random variable, as detailed in [Appendix A, Eq. (29), 1]. After some manipulations,  the desired result is obtained as~\eqref{eq:4ndmoment_hgmjeu_jj}.
\end{proof}

\vspace{-1em}
\begin{Lemma}~\label{lemma2}
    For a random vector $\qx \sim \CN(\boldsymbol{0},\alpha \qI_{M})$ and a projection matrix $\qB = \qI_{M} - \qR (\qR^{\dag} \qR  )^{-1} \qR^{\dag}$ where $\qR \in \C^{M \times N}$ and $\Ex\{ \qB \} = \frac{M - N}{M} \qI_{M}$~\cite{Mohamed_projection_2024}, it holds: $\Ex\big\{ \vert \qx^{\dag} \qB \qx \vert^2 \big\} = (M-N) (M - N + 1) \alpha^2$.
\end{Lemma}
\begin{proof}
   Using the fact that $\qx$ is independent of $\qB$, we have 
    \begin{align}
        \Ex\big\{ \vert \qx^{\dag} \qB \qx \vert^2 \big\}
        &\stackrel{(a)}{=}
        \alpha^2 \Ex\big\{ \big\vert \trace( \qB) \big\vert^2 \big\}
        + \alpha^2 \Ex\big\{ \trace\big( \qB \qB^{\dag} \big) \big\} 
        \nonumber\\
        &\stackrel{(b)}{=}
        \alpha^2 \Ex\big\{ \big\vert \trace( \qU^{H} \LAMBDA \qU ) \big\vert^2 \big\}
        + \alpha^2  \trace\big( \Ex\big\{ \qB \qB^{\dag}  \big\} \big)  
        \nonumber\\
        &=
        \alpha^2 \Ex\big\{ \big\vert \trace(  \LAMBDA \qU \qU^{H} ) \big\vert^2 \big\}
        + \alpha^2  \trace\big( \Ex\big\{ \qB \big\} \big) 
        \nonumber\\ 
        &=
        \alpha^2 \bigg( \sum\nolimits^{M}_{m=1}{\lambda_{m}}  \bigg)^{2}
        + \alpha^2  \trace\bigg( \frac{M-N}{M} \qI_{M} \bigg) 
        \nonumber\\
        &=
        \big( (M-N)^2 + (M - N) \big) \alpha^{2},
    \end{align}
    where (a) follows \cite[Lemma 4]{cite:Chien:TWC:2022} for a random $\qx$ and a given $\qB$. In (b), we decompose  $\qB = \qU \boldsymbol{\Lambda} \qU^{H}$, where $\qU \in \C^{M \times M}$ is a complex unitary matrix and $\boldsymbol{\Lambda}$ is a diagonal matrix containing $M$ eigenvalues, whose value 
    is $\lambda_{m} = \frac{M - N}{M}$. The proof is obtained with further algebraic manipulations.
\end{proof}

\section{Proof of Theorem 1}~\label{appendix:B}
We note that the derivations for $\mathrm{DS}_k$, $\mathrm{BU}_k$ and $\mathrm{EUI}_{kj}$ in~\eqref{eq:SINE:general} follow a similar procedure as in~\cite{Mohammadi:JSAC:2023}. Due to the consideration of PC, the computation of \eqref{eq:IU_interference} adopts different approach. Therefore, we outline the proof of terms related to PC. To this end, we have
\begin{align}~\label{eq:IUI_kkp}
    &\sum\nolimits_{k'\in \K \setminus k} \Ex \vert \mathrm{IUI}_{kk'} \vert^2 
    \nonumber\\
    &= 
    \sum\nolimits_{k'\in \K \setminus k} \Ex \Big\{ \Big\vert  \sum\nolimits_{m \in \M} \sqrt{\rho_d a_m \etamkpI}   (\hgmkue)^\dag \! \wimkp^{\PZF} \Big\vert^{2}  \Big\} 
    \nonumber\\
    &\hspace{0em}
    +\sum\nolimits_{k'\in \K \setminus k} \Ex \Big\{  \Big\vert  \sum\nolimits_{m \in \M}  \sqrt{\rho_d a_m  \etamkpI}  (\gtilmkue)^\dag \wimkp^{\PZF} \Big\vert^{2}  \Big\} 
    \nonumber
\end{align}
\begin{align}
    &\stackrel{(a)}{=} 
     \sum\nolimits_{k'\in \Pir \setminus k}  \Big(  \sum\nolimits_{m \in \M} \sqrt{\rho_d a_m \! \etamkpI (L-K) \gamuemk} \Big)^{\!2} 
     \nonumber\\
      &
     +\sum\nolimits_{k'\in \K \setminus k}  \sum\nolimits_{m \in \M} \rho_d a_m  \etamkpI (\betamkue - \gamuemk),
\end{align}
where (a) follows from that ZF suppresses interference towards all the IRs unless they share the same pilots, i.e.,
\begin{align}~\label{eq:eq76}
    &(\hgmkue)^\dag \wimkp^{\PZF} 
    \nonumber\\
    &= \sqrt{(L-K)\gamuemk}(\qe_{k}^{\mathrm{I}})^\dag \big(\Ghmu\big)^\dag \Ghmu \Big(  \big(\Ghmu\big)^\dag \Ghmu \Big)^{-1}\qe_{k'}^{\mathrm{I}}  \nonumber\\
    &=
    \begin{cases}
      0 & k'\notin\Pir\\
      \sqrt{(L-\tau_{\mathcal{K}})\gamuemk} & k'\in\Pir.
    \end{cases}
\end{align}
Furthermore, due to the statistical independence between $\gtilmkue$ and normalized $\wimk^{\PZF}$, we compute $\Ex \big\{\big\vert (\gtilmkue)^\dag  \wimkp^{\PZF} \big\vert^2\big\} $ as:
\begin{align}~\label{eq:eq77}
\Ex \Big\{(\wimkp^{\PZF})^\dag \gtilmkue(\gtilmkue)^\dag  \wimkp^{\PZF}\Big\}
= \big(\betamkue - \gamuemk \big). 
\end{align}

\section{Proof of Theorem 2}~\label{appendix:C}
The proof of Theorem~2 applies a similar methodology as in Appendix \ref{appendix:B} for Theorem \ref{Theorem:SE:PPZF}, with the difference lying in the analysis of the expectation terms in \eqref{eq:El_average}. 
Considering~\eqref{eq:El_average}, we first derive $\Ex \Big\{\! \big\vert \big(\hgmjue\big)^\dag \!\wemjp^{\PMRT} \big\vert^2 \!\Big\}$ as
\begin{align}~\label{eq:Exgmjwemj}
\Ex \big\{ \big\vert \big(\gmjue \big)^\dag \!\wemj^{\PMRT} \big\vert^2 \!\big\}  \!=\!
 \begin{cases} \frac{1}{\alpha_{mj}^{\PMRT}} 
             \big(\Psi_1 \!+ \!\Psi_2 \big), &\!\!\!\!\!\! j \!= j'\\  
 \frac{1}{\alpha_{mj'}^{\PMRT}}
            \! \Big( \!
            \Phi_1 \!+\!
            \Phi_2
            \Big), &\!\!\!\!\!\!  j,j' \in \Per, j\neq j'.
 \end{cases}
\end{align}
In what follows, we provide expressions and derivations of $\Psi_1$, $\Psi_2$, $\Phi_1$, and $\Phi_2$. Noticing that $\hgmjue \sim \CN( \sqrt{\bar{\zeta}_{mj} \kappa} \gmjeulos, \gameumj \qI_{L})$, $\Psi_1$ can be obtained as
\begin{align}
    \Psi_1 
    &=
    \Ex\bigg\{ \!
    \bigg\vert\! \Big( \sqrt{ \bar{\zeta}_{mj} \kappa} \big(\gmjeulos \big)^\dag + \big( \tildehgmjue \big)^\dag \Big)  \qB_m
    \nonumber\\
    &\hspace{6em}\times\Big(  \sqrt{\bar{\zeta}_{mj} \kappa} \gmjeulos + \tildehgmjue \Big)
    \bigg\vert^2 
    \bigg\}
    \nonumber
\end{align}
\begin{align}
    &=\!
    (\bar{\zeta}_{mj} \kappa)^2 \Ex \Big\{ \! \Big\vert
    (\gmjeulos)^\dag \qB_{m}^\dag \gmjeulos \Big\vert^{2} \!
     \Big\}
     \!\!+\!\!
    \Ex \Big\{ \! \Big\vert
    (\tildehgmjue)^\dag \qB_{m}^\dag \tildehgmjue \Big\vert^{2}
    \! \Big\}
    \nonumber
\end{align}
\begin{align}
    & \hspace{1em} \!+\!
    \bar{\zeta}_{mj} \kappa 
    \Big( 
    \big(\gmjeulos \big)^\dag \Ex\Big\{ \qB_{m}^\dag 
        \Ex\big\{ \tildehgmjue (\tildehgmjue)^\dag \big\}
    \qB_{m} 
    \gmjeulos
    \Big\}
    \Big)
    \nonumber\\
    & \!+\!
    \bar{\zeta}_{mj} \kappa
    \Ex\Big\{ \trace\Big(
        \gmjeulos (\gmjeulos)^\dag \qB_{m}  \tildehgmjue (\tildehgmjue)^\dag \qB_{m}^\dag
    \Big)
    \Big\}
    \nonumber\\
    &  \!+\!
    2\bar{\zeta}_{mj} \kappa
    \Ex\Big\{
    \big(\gmjeulos \big)^\dag \qB_m \gmjeulos 
    \Big\} 
    \Ex\Big\{
    (\tildehgmjue)^\dag \qB_m \tildehgmjue 
    \Big\},      
\end{align}
where we have exploited the statistical independence between $\gmjeulos$ and $\tildehgmjue$ and $\trace(\qx^\dag \qy) = \trace(\qy\qx^\dag)$ for the fourth term. Now, by invoking Lemma~\ref{lemma2} for a random variable $\tildehgmjue \sim \CN(\boldsymbol{0}, \gameumj \qI_{L})$ and a given deterministic matrix $\qB_m \in \C^{L \times L}$, we have
\begin{align}
    \Psi_1     &{\approx}
    (\bar{\zeta}_{mj} \kappa)^2
    \Ex \Big\{ 
    (\gmjeulos)^\dag \qB_{m}^\dag \gmjeulos 
    \Big\}
    \Ex \Big\{
    (\gmjeulos)^\dag \qB_m \gmjeulos 
    \Big\}
    \nonumber\\
    & \hspace{1em} \!+\!
    \big( \gameumj \big)^{\!2}
    \Big( 
    \big(L \!-\! K \big)^{2} + \big(L \!-\! K \big)
    \Big)
    \nonumber\\
    & \hspace{1em} \!+
    \frac{L-K}{L}
    \bar{\zeta}_{mj} \kappa     
        \big(\betamjeu \!+\! N\bar{\zeta}_{mj} \!-\! \gameumj \big)
   (\gmjeulos)^\dag \gmjeulos 
    \nonumber\\
    & \hspace{1em} +
    \frac{L-K}{L}\bar{\zeta}_{mj} \kappa
    \big(\betamjeu + N\bar{\zeta}_{mj} - \gameumj \big)
     \trace\Big( \gmjeulos (\gmjeulos)^\dag \Big)
    \nonumber\\
    & \hspace{1em} \!+\!
    2\bar{\zeta}_{mj} \kappa \frac{L-K}{L} (\gmjeulos)^\dag \gmjeulos (L-K)
     \gameumj,
\end{align}
where the first and the sixth terms exploit the approximation $\Ex\{x^2\} \approx (\Ex\{x\})^2$. To this end, after some manipulation, the result in~\eqref{eq:Psi_1} is obtained. Following similar steps, we can derive $\Phi_1$ as follows: 
\begin{align}~\label{eq:Phi_1}
    &\Phi_1 
    \!\approx\!
    \bigg(\frac{L-K}{L} \!\bigg)^{\!2} \!\! \Big(\kappa^{2} \bar{\zeta}_{m,j} \bar{\zeta}_{m,j'} \!\Big) \big\Vert \gmjeulos \big\Vert^{2} \big\Vert \gmjpeulos \big\Vert^{2}
    \nonumber\\
    & \hspace{1em} \!+
    \big(L\!-\!K \big) \big(L\!-K+\!1 \big) \Bigg(\! \gameumj \frac{ \betamjpeu \!+\! N \bar{\zeta}_{mj'} }{ \betamjeu \!+\! N \bar{\zeta}_{mj} } \!\Bigg)^{\!2}
    \nonumber\\
    & \hspace{1em} +
    \bigg(\frac{L-K}{L} \bigg)\! \Big( N \bar{\zeta}_{mj'} + \betamjpeu - \gameumjp \Big)\bar{\zeta}_{mj} \kappa \big\Vert \gmjeulos \big\Vert^{2}
    \nonumber\\
    & \hspace{1em} +
    \bigg(\frac{L-K}{L} \bigg) \Big(N\bar{\zeta}_{mj} +\betamjeu - \gameumj \Big)\bar{\zeta}_{mj'}\kappa \big\Vert \gmjpeulos \big\Vert^{2}
    \nonumber\\
    & \hspace{1em} \!+\!
    \frac{2(L\!-\!K)^2}{L} 
    \Bigg(\!\! \gameumj \frac{ \betamjpeu \!+\! N \bar{\zeta}_{mj'} }{ \betamjeu \!+\! N \bar{\zeta}_{mj} } \!\Bigg)
    \sqrt{\!\kappa^{2} \bar{\zeta}_{m,j} \bar{\zeta}_{m,j'}\!}
    \nonumber\\
    & \hspace{2em} \!\times\! \Big( (\gmjpeulos)^{H} \gmjeulos + (\gmjeulos)^{H} \gmjeulos \Big).
\end{align} 
By exploiting Lemma~1 and the independence between the statistics of estimated channel $\hgmjue$, channel estimation error, $\gtilmjeu$, and projection matrix, $\qB_m$, we can obtain $\Psi_2$ as
\begin{align}
    &\Ex \Big\{\! \trace \Big(
        \hgmjue \big(\hgmjue \big)^\dag \qB_m \gtilmjeu \big( \gtilmjeu \big)^\dag \qB_m^\dag \Big)
        \Big\} 
    \\
    &\hspace{0.5em} =
    \frac{L - K}{L}
        \Big(\betamjeu \!+\! N\bar{\zeta}_{mj} - \gameumj \Big) \Big( L\gameumjp\! +\! \bar{\zeta}_{mj'} \kappa \big\Vert \gmjpeulos \big\Vert^2 \Big). \nonumber
\end{align}
Since $\Phi_2$ has the same properties as $\Phi_1$, similar algebraic manipulations can be applied.

Next, we compute the third expectation term in~\eqref{eq:El_average}, for the case $\forall j, j' \in \mathcal{P}_x$, as
\begin{align}~\label{eq:thirdExp}
    \Ex \Big\{ \Big\vert \big(\gmjue\big)^\dag\wimk^{\PZF}  \Big\vert^2  \Big\} 
    &\!\!\stackrel{(a)}{=}\!\!  \Ex\Big\{ \big(\gmjue\big)^{\!\dag} \Ex\Big\{\! (\wimk^{\PZF})^\dag \wimk^{\PZF} \!\Big\}  \gmjue \!\Big\} \nonumber\\
    &=\! \betamjeu 
        \!+\! N \bar{\zeta}_{mj} \!+\! \frac{\bar{\zeta}_{mj}}{L} \kappa \big\Vert \gmjeulos \big\Vert^{2}, 
\end{align}
where (a) exploits that $\gmjue$ is independent of $\wimk^{\PZF}$ for ER $j$ and IR $k$ and Lemma~1.  To this end, by substituting the final results for~\eqref{eq:Exgmjwemj} and~\eqref{eq:thirdExp} into~\eqref{eq:El_average}, after some manipulations, we get the desired result in~\eqref{eq:El_average:PPZF}. 


\begin{thebibliography}{10}
\providecommand{\url}[1]{#1}
\csname url@samestyle\endcsname
\providecommand{\newblock}{\relax}
\providecommand{\bibinfo}[2]{#2}
\providecommand{\BIBentrySTDinterwordspacing}{\spaceskip=0pt\relax}
\providecommand{\BIBentryALTinterwordstretchfactor}{4}
\providecommand{\BIBentryALTinterwordspacing}{\spaceskip=\fontdimen2\font plus
\BIBentryALTinterwordstretchfactor\fontdimen3\font minus \fontdimen4\font\relax}
\providecommand{\BIBforeignlanguage}[2]{{%
\expandafter\ifx\csname l@#1\endcsname\relax
\typeout{** WARNING: IEEEtran.bst: No hyphenation pattern has been}%
\typeout{** loaded for the language `#1'. Using the pattern for}%
\typeout{** the default language instead.}%
\else
\language=\csname l@#1\endcsname
\fi
#2}}
\providecommand{\BIBdecl}{\relax}
\BIBdecl

\bibitem{Hua:WCNC:2024}
T.~D. Hua, M.~Mohammadi, H.~Q. Ngo, and M.~Matthaiou, ``Cell-free massive {MIMO SWIPT} with beyond diagonal reconfigurable intelligent surfaces,'' in \emph{Proc. IEEE WCNC}, Apr. 2024, pp. 1--6.

\bibitem{Alsaba:TuT:2018}
Y.~Alsaba, S.~K.~A. Rahim, and C.~Y. Leow, ``Beamforming in wireless energy harvesting communications systems: A survey,'' \emph{{IEEE} Commun. Surveys Tuts.}, vol.~20, no.~2, pp. 1329--1360, Secondquarter 2018.

\bibitem{cite:HienNgo:cf02:2018}
H.~Q. Ngo, L.-N. Tran, T.~Q. Duong, M.~Matthaiou, and E.~G. Larsson, ``On the total energy efficiency of cell-free massive {MIMO},'' \emph{{IEEE} Trans. Green Commun. Netw.}, vol.~2, no.~1, pp. 25--39, Mar. 2018.

\bibitem{mohammadi2024next}
M.~Mohammadi, Z.~Mobini, H.~Q. Ngo, and M.~Matthaiou, ``Next-generation multiple access with cell-free massive {MIMO},'' \emph{Proc. {IEEE}}, vol. 112, no.~9, pp. 1372--1420, Sept. 2024.

\bibitem{cite:reviewer3a}
E.~Shi \emph{et~al.}, ``{RIS}-aided cell-free massive {MIMO} systems for {6G}: Fundamentals, system design, and applications,'' \emph{Proc. IEEE}, vol. 112, no.~4, pp. 331--364, Jun. 2024.

\bibitem{cite:reviewer3b}
------, ``Wireless energy transfer in {RIS}-aided cell-free massive {MIMO} systems: Opportunities and challenges,'' \emph{{IEEE} Commun. Mag.}, vol.~60, no.~3, pp. 26--32, Mar. 2022.

\bibitem{Kaixi:2024:WCL}
K.~Yang, Y.~Zhang, L.~Chen, C.~Tang, H.~Fang, and L.~Yang, ``On the performance of {RIS}-aided {CF}-{IoE}-{SWIPT} network over correlated {Rician} fading channels,'' \emph{{IEEE} Wireless Commun. Lett.}, vol.~13, no.~10, pp. 2692--2696, Oct. 2024.

\bibitem{Kaixi:2021:China}
Z.~Yang and Y.~Zhang, ``Beamforming optimization for {RIS}-aided {SWIPT} in cell-free {MIMO} networks,'' \emph{China Commun.}, vol.~18, no.~9, pp. 175--191, Sept. 2021.

\bibitem{Mohammadi:TC:2024}
M.~Mohammadi, H.~Q. Ngo, and M.~Matthaiou, ``Phase-shift and transmit power optimization for {RIS}-aided massive {MIMO} {SWIPT} {IoT} networks,'' \emph{{IEEE} Trans. Commun.}, vol.~73, no.~1, pp. 631--647, Jan. 2025.

\bibitem{cite:R1}
K.~An \emph{et~al.}, ``Exploiting multi-layer refracting ris-assisted receiver for hap-swipt networks,'' \emph{{IEEE} Trans. Wireless Commun.}, vol.~23, no.~10, pp. 12\,638--12\,657, Oct. 2024.

\bibitem{cite:HongyuLi:BDRISoverview01:2023}
H.~Li, S.~Shen, M.~Nerini, and B.~Clerckx, ``Reconfigurable intelligent surfaces 2.0: Beyond diagonal phase shift matrices,'' \emph{{IEEE} Commun. Mag.}, vol.~62, no.~3, pp. 102--108, Mar. 2024.

\bibitem{Li:JSAC:2023}
H.~Li, S.~Shen, and B.~Clerckx, ``Beyond diagonal reconfigurable intelligent surfaces: A multi-sector mode enabling highly directional full-space wireless coverage,'' \emph{{IEEE} J. Sel. Areas Commun.}, vol.~41, no.~8, pp. 2446--2460, Aug. 2023.

\bibitem{Wang:JIOT:2020}
X.~Wang, A.~Ashikhmin, and X.~Wang, ``Wirelessly powered cell-free {IoT: Analysis} and optimization,'' \emph{{IEEE} Internet Things J.}, vol.~7, no.~9, pp. 8384--8396, Sept. 2020.

\bibitem{Demir:TWC:2021}
O.~T. Demir and E.~Bj\"{o}rnson, ``Joint power control and {LSFD} for wireless-powered cell-free massive {MIMO},'' \emph{{IEEE} Trans. Wireless Commun.}, vol.~20, no.~3, pp. 1756--1769, Mar. 2021.

\bibitem{Femenias:TCOM:2021}
G.~Femenias, J.~García-Morales, and F.~Riera-Palou, ``{SWIPT}-enhanced cell-free massive {MIMO} networks,'' \emph{{IEEE} Trans. Commun.}, vol.~69, no.~8, pp. 5593--5607, Aug. 2021.

\bibitem{Xinjiang:TWC:2021}
X.~Xia, P.~Zhu, J.~Li, H.~Wu, D.~Wang, Y.~Xin, and X.~You, ``Joint user selection and transceiver design for cell-free with network-assisted full duplexing,'' \emph{{IEEE} Trans. Wireless Commun.}, vol.~20, no.~12, pp. 7856--7870, Dec. 2021.

\bibitem{Zhang:IoT:2022}
Y.~Zhang, W.~Xia, H.~Zhao, W.~Xu, K.-K. Wong, and L.~Yang, ``Cell-free {IoT} networks with {SWIPT: Performance} analysis and power control,'' \emph{{IEEE} Internet Things J.}, vol.~9, no.~15, pp. 13\,780--13\,793, Aug. 2022.

\bibitem{Galappaththige:WCL:2024}
D.~Galappaththige and C.~Tellambura, ``Sum rate maximization for {RSMA}-assisted {CF mMIMO} networks with {SWIPT} users,'' \emph{{IEEE} Wireless Commun. Lett.}, vol.~13, no.~5, pp. 1300--1304, May 2024.

\bibitem{Zhang:TWC:2023}
R.~Zhang, K.~Xiong, Y.~Lu, D.~W.~K. Ng, P.~Fan, and K.~B. Letaief, ``{SWIPT}-enabled cell-free massive {MIMO-NOMA} networks: A machine learning-based approach,'' \emph{{IEEE} Trans. Wireless Commun.}, vol.~23, no.~7, pp. 6701--6718, Jul. 2023.

\bibitem{Mohammadi:GC:2023}
M.~Mohammadi, L.-N. Tran, Z.~Mobini, H.~Q. Ngo, and M.~Matthaiou, ``Cell-free massive {MIMO}-assisted {SWIPT} for {IoT} networks,'' \emph{{IEEE} Trans. Wireless Commun.}, pp. 1--1, 2025.

\bibitem{cite:Chien:TWC:2022}
T.~V. Chien, H.~Q. Ngo, S.~Chatzinotas, M.~D. Renzo, and B.~Ottersten, ``Reconfigurable intelligent surface-assisted cell-free massive {MIMO} systems over spatially-correlated channels,'' \emph{{IEEE} Trans. Wireless Commun.}, vol.~21, no.~7, pp. 5106--5128, Jul. 2022.

\bibitem{Lan:IoT:2024}
M.~Lan, Y.~Hei, M.~Huo, H.~Li, and W.~Li, ``A new framework of {RIS}-aided user-centric cell-free massive {MIMO} system for {IoT} networks,'' \emph{{IEEE} Internet Things J.}, vol.~11, no.~1, pp. 1110--1121, Jan. 2024.

\bibitem{Elhoushy:WCL:2022}
S.~Elhoushy, M.~Ibrahim, and W.~Hamouda, ``Exploiting {RIS} for limiting information leakage to active eavesdropper in cell-free massive {MIMO},'' \emph{{IEEE} Wireless Commun. Lett.}, vol.~11, no.~3, pp. 443--447, Mar. 2022.

\bibitem{cite:b6}
\BIBentryALTinterwordspacing
A.~Azarbahram \emph{et~al.}, ``Beamforming and waveform optimization for {RF} wireless power transfer with beyond diagonal reconfigurable intelligent surfaces,'' 2025. [Online]. Available: \url{https://arxiv.org/abs/2502.19176}
\BIBentrySTDinterwordspacing

\bibitem{cite:Shen:BDRISoverview02:2022}
S.~Shen, B.~Clerckx, and R.~Murch, ``Modeling and architecture design of reconfigurable intelligent surfaces using scattering parameter network analysis,'' \emph{{IEEE} Trans. Wireless Commun.}, vol.~21, no.~2, pp. 1229--1243, Feb. 2022.

\bibitem{Hien:TDD:pilot}
H.~Q. Ngo and E.~G. Larsson, ``No downlink pilots are needed in {TDD} massive {MIMO},'' \emph{{IEEE} Trans. Wireless Commun.}, vol.~16, no.~5, pp. 2921--2935, 2017.

\bibitem{cite:zahra:statisticalCSI}
\BIBentryALTinterwordspacing
Z.~Mobini and H.~Q. Ngo, ``Massive {MIMO}: Instantaneous versus statistical {CSI}-based power allocation,'' 2025. [Online]. Available: \url{https://arxiv.org/abs/2505.04294}
\BIBentrySTDinterwordspacing

\bibitem{cite:HienNgo:cf01:2017}
H.~Q. Ngo, A.~Ashikhmin, H.~Yang, E.~G. Larsson, and T.~L. Marzetta, ``Cell-free massive {MIMO} versus small cells,'' \emph{{IEEE} Trans. Wireless Commun.}, vol.~16, no.~3, pp. 1834--1850, Mar. 2017.

\bibitem{cite:book:Kay:1993}
S.~M. Kay, \emph{Fundamentals of Statistical Signal Processing: Estimation Theory}.\hskip 1em plus 0.5em minus 0.4em\relax USA: Prentice-Hall, Inc., Mar. 1993.

\bibitem{cite:MRT_for_HE:Almradi}
A.~Almradi and K.~A. Hamdi, ``The performance of wireless powered {MIMO} relaying with energy beamforming,'' \emph{{IEEE} Trans. Commun.}, vol.~64, no.~11, pp. 4550--4562, Nov. 2016.

\bibitem{cite:pzf_pmrt_2020}
G.~Interdonato, M.~Karlsson, E.~Björnson, and E.~G. Larsson, ``Local partial zero-forcing precoding for cell-free massive {MIMO},'' \emph{{IEEE} Trans. Wireless Commun.}, vol.~19, no.~7, pp. 4758--4774, Jul. 2020.

\bibitem{Boshkovska:CLET:2015}
E.~Boshkovska, D.~W.~K. Ng, N.~Zlatanov, and R.~Schober, ``Practical non-linear energy harvesting model and resource allocation for {SWIPT} systems,'' \emph{{IEEE} Commun. Lett.}, vol.~19, no.~12, pp. 2082--2085, Dec. 2015.

\bibitem{Che:TWC:2014}
E.~Che, H.~D. Tuan, and H.~H. Nguyen, ``Joint optimization of cooperative beamforming and relay assignment in multi-user wireless relay networks,'' \emph{{IEEE} Trans. Wireless Commun.}, vol.~13, no.~10, pp. 5481--5495, Oct. 2014.

\bibitem{Mohammadi:JSAC:2023}
M.~Mohammadi, T.~T. Vu, H.~Q. Ngo, and M.~Matthaiou, ``Network-assisted full-duplex cell-free massive {MIMO}: Spectral and energy efficiencies,'' \emph{{IEEE} J. Sel. Areas Commun.}, vol.~41, no.~9, pp. 2833--2851, Sept. 2023.

\bibitem{cite:Grant:CVX}
M.~Grant and B.~SP, ``{CVX}: {MATLAB} software for disciplined convex programming,'' Jan. 2014.

\bibitem{vu18TCOM}
T.~T. {Vu}, D.~T. {Ngo}, M.~N. {Dao}, S.~{Durrani}, and R.~H. {Middleton}, ``Spectral and energy efficiency maximization for content-centric {C-RANs} with edge caching,'' \emph{IEEE Trans. Commun.}, vol.~66, no.~12, pp. 6628--6642, Dec. 2018.

\bibitem{tam16TWC}
H.~H.~M. Tam, H.~D. Tuan, D.~T. Ngo, T.~Q. Duong, and H.~V. Poor, ``Joint load balancing and interference management for small-cell heterogeneous networks with limited backhaul capacity,'' \emph{{IEEE} Trans. Commun.}, vol.~16, no.~2, pp. 872--884, Feb. 2017.

\bibitem{Thien:DDPG:2023}
D.-T. Hua, Q.~T. Do, N.-N. Dao, T.-V. Nguyen, D.~Shumeye~Lakew, and S.~Cho, ``Learning-based reconfigurable-intelligent-surface-aided rate-splitting multiple access networks,'' \emph{{IEEE} Internet Things J.}, vol.~10, no.~20, pp. 17\,603--17\,619, Oct. 2023.

\bibitem{YMao_BDRIS_2023}
T.~Fang and Y.~Mao, ``A low-complexity beamforming design for beyond-diagonal {RIS} aided multi-user networks,'' \emph{IEEE Commun. Lett.}, vol.~28, no.~1, pp. 203--207, Jan. 2024.

\bibitem{hydrid:RIS}
A.~Huang, X.~Mu, L.~Guo, and G.~Zhu, ``Energy-efficient design for hybrid {RIS} transmitter enabled multi-user communications,'' in \emph{Proc. IEEE WCNC}, Apr. 2024, pp. 1--6.

\bibitem{cite:hua2024BdrisSwipt}
\BIBentryALTinterwordspacing
T.~D. Hua, M.~Mohammadi, H.~Q. Ngo, and M.~Matthaiou, ``{SWIPT} in cell-free massive {MIMO} using stacked intelligent metasurfaces,'' in \emph{Proc. IEEE ICC}, Jun. 2025. [Online]. Available: \url{https://arxiv.org/abs/2503.14032}
\BIBentrySTDinterwordspacing

\bibitem{Mohamed_projection_2024}
M.~Elfiatoure, M.~Mohammadi, H.~Q. Ngo, P.~J. Smith, and M.~Matthaiou, ``Protecting massive {MIMO}-radar coexistence: Precoding design and power control,'' \emph{{IEEE} Open J. Commun. Society}, vol.~5, pp. 276--293, Jan. 2024.

\end{thebibliography}
\end{document}